\documentclass[12pt]{article}
\usepackage[top=1.25in,bottom=1in,left=1in,right=1in]{geometry}
\usepackage[utf8]{inputenc}
\usepackage{amsmath}
\usepackage{natbib}
\usepackage[at]{easylist}
\usepackage{ascmac} 
\usepackage{amssymb}
\usepackage{amsthm}
\usepackage{bbm}
\usepackage{xcolor}
\usepackage{graphicx}
\usepackage{float}
\usepackage{enumerate}
\usepackage{natbib}
\usepackage{mathrsfs}
\usepackage{comment}
\usepackage{apptools}

\newcommand{\RN}[1]{  \textup{\uppercase\expandafter{\romannumeral#1}}}

\newtheorem{definition}{Definition}
\newtheorem{proposition}{Proposition}
\newtheorem{theorem}{Theorem}
\newtheorem{corollary}{Corollary}
\newtheorem{lemma}{Lemma}
\newtheorem{sublemma}{Sublemma}

\newtheorem{assumption}{Assumption}

\theoremstyle{remark}
\newtheorem{example}{Example}

\makeatother

\newcommand{\R}{\mathbb{R}}
\newcommand{\E}{\mathbb{E}}
\renewcommand{\P}{\mathbb{P}}

\newcommand{\Var}{\text{Var}}

\newcommand{\maxx}[1]{\max\{#1\}}

\AtAppendix{
\numberwithin{equation}{section}
\numberwithin{definition}{section}
\numberwithin{theorem}{section}
\numberwithin{proposition}{section}
\numberwithin{lemma}{section}}

\DeclareMathOperator*{\argmax}{arg\,max}

\newcommand{\indep}{\perp\!\!\!\perp}

\allowdisplaybreaks

\title{The Value of Context: \\
Human versus Black Box Evaluators}
\author{Andrei Iakovlev and Annie Liang\thanks{Department of Economics, Northwestern University. We thank Modibo Camara, Krishna Dasaratha, Alex Frankel, Ben Golub, Kevin He, Xiaosheng Mu, Matthew Murphy, Jacopo Perego, Debraj Ray, and Marzena Rostek for helpful comments and suggestions.}}
\date{\today}
\begin{document}

\maketitle

\begin{abstract}
Machine learning algorithms are now capable of performing evaluations previously conducted by human experts (e.g., medical diagnoses). How should we conceptualize the difference between evaluation by humans and by algorithms, and when should an individual prefer one over the other? We propose a framework to examine one key distinction between the two forms of evaluation: Machine learning algorithms are standardized, fixing a common set of covariates by which to assess all individuals, while human evaluators customize which covariates are acquired to each individual. Our framework defines and analyzes the advantage of this customization---the \emph{value of context}---in environments with high-dimensional data. We show that unless the agent has precise knowledge about the joint distribution of covariates, the benefit of additional covariates generally outweighs the value of context. 
\end{abstract}

\section{Introduction}

\begin{quote} \footnotesize{``A statistical formula may be highly successful in predicting whether or not a person will go to a movie in the next week. But someone who knows that this person is laid up with a broken leg will beat the formula. No formula can take into account the infinite range of such exceptional events.''
\hfill --- Atul Gawande, \emph{Complications: A Surgeon's Notes on an Imperfect Science} }
\end{quote}

\normalsize
\bigskip

Predictions about people are increasingly automated using black-box algorithms. How should individuals compare evaluation by algorithms (e.g., medical diagnosis by a machine learning algorithm) with more traditional evaluation by human experts (e.g., medical diagnosis by a doctor)?

 One important distinction is  that black-box algorithms are standardized, fixing a common set of inputs by which to assess all individuals. Unless the inputs to the black box are  exhaustive, additional information can (in some cases) substantially modify the interpretation of those inputs that have been acquired. For example,  the context that a patient is currently fasting may change the interpretations of  ``dizziness'' and ``electrolyte imbalance,'' and the context that a job applicant is an environmental activist may change how a prior history of arrest is perceived. If these auxiliary characteristics are not specified as inputs in the algorithm, the individual cannot supply them.
 
 In contrast, individuals can often explain their unusual circumstances or characteristics to a human evaluator through conversation. Thus, even if the human evaluator  considers fewer inputs than a black box algorithm does, these inputs  may be better adapted to the individual being evaluated. The perception that humans are better able to take into account an individual's unique situation is a significant factor in  patient resistance to AI in healthcare \citep{Longonietal2019}. Our objective is to understand when, and to what extent, this difference between human and machine evaluation matters.

 Our contribution in this paper is twofold. First, we propose a theoretical framework that formalizes this distinction between human and black box evaluation. Second, we explain assumptions under which it will turn out that the the agent should prefer one form of evaluation over the other. We see our paper as a complement to a growing empirical literature that compares human versus black box evaluation. Here our goal is to conceptualize the difference between human and black box evaluators, and to clarify properties of the informational environment that are important for choosing between the two.

In our model, an agent is described by a binary covariate vector and a real-valued type (e.g., the severity of the agent's medical condition). The type can be written as a function of the covariates, which we henceforth call the \emph{type function}. Covariates are separated into standard covariates (e.g., medical history, lab tests, imaging scans) and nonstandard covariates (e.g., religious information, genetic data, wearable device data, and financial data). 

We suppose that the agent may know how the standard covariates are correlated with the type, but cannot distinguish between the predictive roles of the nonstandard covariates. Formally, the agent has a belief over the type function, and we impose two assumptions on the agent's prior. The first is a symmetry assumption  that says that the agent's prior over these functions is unchanged by permuting the labels and values of the nonstandard covariates. If we interpret the covariates as signals about the agent's type, then uncertainty about the type function corresponds to uncertainty about the signal structure (\`{a} la model uncertainty, e.g., \citet{ACY2015} and \citet{MorrisYildiz}). The second assumption fixes the unpredictability of the agent's type to be constant in the total number of covariates. We impose this because in many applications, machine learning algorithms have millions of inputs, and yet cannot predict the outcome perfectly. Thus our ``many covariates'' limit does not represent a situation in which the total of amount of information grows large, but rather one in which the type function can be arbitrarily complex. We view these two assumptions as useful conceptual benchmarks, but subsequently show that neither are essential for our main results (see Section \ref{sec:ExtendBreak} for details).

The agent's payoff is determined by his true type and an evaluation, which may be made either by a human evaluator or a black-box evaluator.
In either case the evaluation is a conditional expectation of the agent's type given the agent's standard covariates and some fraction of the agent's nonstandard covariates. But the sets of nonstandard covariates that are observed by the black box evaluator and the human evaluator differ in two ways.

First, the black box evaluator observes a larger fraction of the nonstandard covariates than the human evaluator does (since humans cannot process as much information as black box algorithms can). Second, the nonstandard covariates observed by the black box evaluator are a pre-specified set of algorithmic inputs, which  are fixed across individuals. For example, a designer of a medical algorithm may specify a set of inputs including (among others) blood type, BMI, and smoking status. The black box algorithm learns a mapping from those inputs into the diagnosis.  We view the human evaluator as instead uncovering nonstandard covariates during a conversation, where the specific path of questioning may vary across agents. Thus the human evaluator may end up learning about one individual's sleep schedule but another individual's financial situation.

Rather than modeling these conversations directly, we consider an upper bound on the agent's payoff under human evaluation, where the covariates that the human observes are the ones that maximize the agent's payoffs (subject to the human's capacity constraint). We say that the agent prefers the black box if  the agent's expected payoffs are higher under black box evaluation even compared to these \emph{best-case} conversations with the human.

This comparison essentially reduces to the question of whether the agent prefers an evaluator who observes a larger fraction of (non-targeted) nonstandard covariates about the agent, or an evaluator who observes a smaller but targeted fraction of nonstandard covariates.\footnote{This question is spiritually related to \citet{AkbarpourMalladiSaberi}'s comparison of the network diffusion value of a small number of targeted seeds versus a larger number of randomly selected seeds. Like them, we will find that a larger number of (non-targeted) inputs is superior, but the mechanisms behind these results are very different; in particular, network structure does not play a role in our results.} Towards this comparison, we first introduce a  benchmark, which is the expected payoff that the agent would receive if interacting with an evaluator who observes no nonstandard covariates. 
We define the \emph{value of context} to be the improvement in the agent's payoffs under best-case human evaluation, relative to this benchmark.  The value of context thus quantifies the  extent to which the agent's payoffs can possibly be improved when the evaluator observes nonstandard covariates suited to that agent.

 Our first main result says that under our assumptions on the agent's prior, the expected value of context vanishes to zero as the number of covariates grows large.  
Thus even though there may be realizations of the type function given which the value of context is large, in \emph{expectation} it is not.   The contrapositive of this result is that if the expected value of context is high in some application, it must be that our assumptions on the prior do not hold, i.e., the agent has some ex-ante knowledge about the predictive roles of the nonstandard covariates. 

We prove this result by studying the sensitivity of the evaluator's expectation to the set of covariates that are revealed. Intuitively, a large value of context requires that the evaluator's beliefs move sharply after observing certain nonstandard covariates. We show that the largest feasible change in the evaluator's beliefs can be written as the maximum over a set of random variables, each corresponding to the movement in the evaluator's beliefs for a given choice of covariates to reveal. The proof proceeds by first reducing this problem to studying the maximum of a growing sequence of (appropriately constructed) i.i.d.\ variables, and  then applying a result from \citet{chernozhukov2013gaussian} to show that this expected maximum concentrates on its expectation as the number of covariates grows large. We conclude by bounding this expectation and demonstrating that it vanishes.

We next use this result to compare the agent's expected payoff under human and black box evaluation, when the total number of covariates is sufficiently large. We show that when the agent prefers a more accurate evaluation---formally, when the agent's payoff is convex in the evaluation---the agent should (eventually) prefer an algorithmic evaluator with access to more covariates over a human evaluator to whom the agent can provide context. And when the agent's payoff is concave in the evaluation, the conclusion is (eventually) reversed. We view these conclusions as relevant not only in the many-covariates limit: We quantify the number of covariates that is needed for our result to hold, and show that it can be quite small. For example, if the agent's utility function satisfies mild regularity conditions, the Human evaluator observes 10\% of covariates, and the Black Box evaluator observes 90\%, of covariates, then our result holds as long as there are at least 14 covariates.

We subsequently strengthen our main results in two ways: First, we show that not only does the expected value of context vanish for each  agent, but in fact the expected \emph{maximum} value of context across agents also vanishes. Thus,  the expected value of context is eventually small for everyone in the population. Second, we  show that our main results extend when the agent and evaluator interact in a disclosure game, where the agent chooses which nonstandard covariates to reveal, and the evaluator makes inferences about the agent based on which covariates are revealed  (given the agent's equilibrium reporting strategy).

We conclude by examining the role of our assumptions about the agent's prior, and the extent to which our results depend on them. First, we study two variations of our main model, in which the symmetry assumption is relaxed: In the first, we suppose that there is a ``low-dimensional'' set of covariates that relevant for predicting the agent's type;  in the second, we suppose that the agent knows ex-ante the predictive role of certain nonstandard covariates. In both of these settings, our main results extend partially but can also fail: For example, if the set of relevant covariates is sufficiently small that they can be fully disclosed to the evaluator, then the expected value of context typically will not vanish. Next we show that our results extend also in a model in which the predictability of the agent's type is higher in environments with a larger number of covariates (thus relaxing our second asumption on the prior). Finally we provide an abstract learning condition under which our results extend: It is enough for the informativeness of each individual set of covariates to vanish as the total number of covariates grows large.  Together with our main results, these extensions clarify different categories of informational assumptions under which the expected value of context does or does not turn out to be high.

Our model is not meant to be a complete description of the differences between human and black box evaluation. For example, we do not consider human or algorithmic bias \citep{Kleinbergetal2017,gillis2021fairness}, explainability \citep{Yangetal}, preferences for empathetic evaluators, or the possibility that the human evaluator has access to information that is not available to the algorithm (e.g., for privacy protection reasons as in \citet{Agarwaletal2023}). We also suppose that both evaluators form correct conditional expectations, thus abstracting away from the possibility of algorithmic overfitting and of bounded human rationality (e.g., as considered in \citet{Spiegler} and \citet{JacksonHaghtalabProcaccia}).\footnote{The problem of overfitting, while practically important, is a function of how the algorithm is trained. We are interested here in intrinsic differences between the qualitative nature of human and black box evaluation, which are difficult to resolve by training the algorithm differently.} We leave extensions of our model that include these other interesting differences to future work.

\subsection{Related Literature}

Our paper is situated at the intersection of the literatures on learning (Section \ref{subsec:Learning}) and strategic information disclosure (Section \ref{subsec:InfoDisclosure}), where our analysis is primarily differentiated from the previous frameworks by our assumption that the agent has model uncertainty (see Section \ref{subsec:Learning}). Our paper is also inspired by a recent empirical literature that compares human and AI evaluation, which we review in Section \ref{subsec:HumanAI}.

\subsubsection{(Asymptotic) Learning} \label{subsec:Learning}

A large literature studies asymptotic learning and agreement across Bayesian agents \citep{BlackwellDubins}. Our main result (Theorem \ref{thm:Main}) can be viewed as bounding (in expectation) the differences in beliefs across Bayesian agents who are given different information. As in \citet{Vives}, \citet{GolubJackson}, \citet{LiangMu2020}, \citet{Hareletal}, and \citet{FrickIijimaIshii} among others, we quantify the rate of convergence in beliefs. The learning rates that we look at are, however, of a different nature from those studied previously. One important distinction is that these previous papers consider asymptotics as the total amount of information accumulates, while our analysis considers asymptotics with respect to a sequence of information structures that we show are increasingly less informative. A second important difference is that the classic learning models suppose that the agent updates to a signal with a known signal structure, while our agent has uncertainty over the signal structure (as in \citet{ACY2015} and \citet{MorrisYildiz}).  Our results characterize the informativeness of this signal in expectation, where the agent's model uncertainty takes a particular (and new) form motivated by the applications we have in mind.

Finally, our paper is related to \citet{DiTillioOttavianiSorensen}, which compares the informativeness of an unbiased signal to the informativeness of a selected signal whose realization is the maximum realization across i.i.d.\ unbiased signals. Again the key difference is our assumption of model uncertainty---that is, in  \citet{DiTillioOttavianiSorensen}, the signal structures that are being compared are deterministic and known, while in ours they are random and compared in expectation. In particular, our agent's prior belief over signal structures can have support on signal processes which are not i.i.d. (for example, it may be that the meaning of one signal is dependent on the meaning of another).

\subsubsection{Strategic Information Disclosure} \label{subsec:InfoDisclosure}

Several literatures study persuasion via strategic information disclosure.  Our model---in which the sender has private information about his type vector, and selectively chooses which elements to disclose to a naive receiver---is closest to models of disclosure of hard information \citep{Dye1985,GrossmanHart1980}, in particular \citet{Milgrom1981}.\footnote{A similar model of information is considered in \cite{glazer2004optimal} and \cite{anticselected}.} The key difference (which follows from our assumption of model uncertainty) is that our sender has uncertainty about how his reports are interpreted. Additionally, our focus is not on examining which incentive-compatible reporting strategy is optimal,\footnote{Indeed, in our main model we do not require choice of an incentive-compatible reporting strategy, since the receiver updates to the sender's disclosure as if it were exogenous information. This is primarily for convenience---we show in Section \ref{sec:Disclosure} that our results extend in a disclosure game.} but instead on asymptotic limits of belief manipulability as the number of components in the type vector grows large. This latter focus is special to our motivating applications.

Our model also has important differences from the other main strands of the persuasion literature. Unlike models of cheap talk \citep{crawford1982strategic}, our agent chooses between messages whose meanings are fixed exogenously (through the realization of the joint distribution relating covariates to the type) rather than in an equilibrium. Unlike the literature on Bayesian persuasion (\cite{KamenicaGentzkow}), our sender chooses which signal realization to share ex-post  from a finite set of signal realizations, rather than committing to a flexibly chosen information structure ex-ante.\footnote{Thus, for example, Bayes plausibility is not satisfied in our setting---the sender's expectation of the receiver's expectation of the state (following disclosure) is generally not the prior expectation of the state.} Indeed, our model gives the sender substantial power to influence the receiver's beliefs relative to this previous literature. It is perhaps surprising, then, that despite the lack of constraints imposed on the sender, we find that the sender is extremely limited in his influence. In our model, this emerges because the sender has a limited choice from a set of information structures, whose informativeness (we show) is vanishing in the total number of covariates.\footnote{The covariates in our model play a similar role to attributes, although the literature on attributes has focused on choice of which attributes to learn about (e.g., \citet{KLABJAN2014190} and \citet{LiangMuSyrgkanis2022}), rather than which attributes to disclose for the purpose of persuasion. An exception is   \citet{Bardhi}, who studies a principal-agent problem in which a principal selectively samples attributes to influence an agent decision.}

\subsubsection{Human vs AI Evaluation} \label{subsec:HumanAI}

Recent empirical papers compare the accuracy of human evaluation with AI evaluation,  finding that machine learning algorithms outperform experts in problems including medical diagnosis \citep{rajpurkar2017chexnet,jung2017simple,Agarwaletal2023}, prediction of pretrial misconduct \citep{Kleinbergetal2017,Angelovaetal}, and prediction of worker productivity \citep{Chalfinetal2016}. Nonetheless, many individuals continue to distrust algorithmic predictions \citep{Jussupowetal2020,bastani2022improving,Laietal2023}. These findings motivate our goal of understanding whether individuals should prefer human evaluators, and when instead the replacement of human evaluation with algorithmic evaluation is welfare-improving for users, as suggested in \citet{ObermeyerEmanuel} among others.

In principle, human decision-making guided by algorithmic predictions should be superior to either human or algorithmic prediction alone. In practice the evidence is more mixed, with the provision of algorithmic recommendations sometimes leading human decision-makers to less accurate predictions \citep{Hoffmanetal2017,Angelovaetal,Agarwaletal2023}.\footnote{Other papers instead consider algorithmic prediction tools that take human evaluation as an input, with greater success towards improving accuracy (e.g.,  \citet{raghu2019algorithmic}).}  The question of how to  aggregate human and machine evaluations is thus important but subtle, and depends on (among other things) whether human decision-makers understand the correlation between their information and that of the algorithm \citep{mclaughlin2022algorithmic,gillis2021fairness,Agarwaletal2023}. We abstract away from these complexities, focusing instead on (one aspect of) the more basic question of why human oversight is even necessary to begin with. We provide a tractable way of formalizing the advantage of human evaluation, and quantify the size of this advantage.

 \section{Model}

\subsection{Setup}
Agents are each described by a binary covariate vector $\bold{x}_n = (x_1, x_2, \dots, x_n)$ and a type $y \in \left[-\overline{y},\overline{y}\right]$ (where $0\le \overline{y} < \infty$), which are structurally related by the function
\[y = f(x_1, \dots,x_n).\]
We refer to $f$ henceforth as the \emph{type function}. The distribution over covariate vectors is uniform in the population.\footnote{All of our results extend for arbitrary finite-valued covariates.}

We refer to the  covariates indexed to $\mathcal{S} = \{1,\dots,s\}$ as \emph{standard} covariates and the covariates indexed to $\mathcal{N} = \{s+1, \dots, n\}$ as \emph{nonstandard} covariates.  

\begin{example}[Job Interview] \label{ex:Loan} Standard covariates describing a job applicant may include their work history, education level, college GPA, and the coding languages they know. Nonstandard covariates may include  their social media activity (e.g., number of followers, posts, likes), wearable device data (e.g., sleep patterns, physical activity level), and hobbies (e.g., whether they are active readers, whether they enjoy extreme sports).\end{example}

\begin{example}[Medical Prediction] \label{ex:Medical} Standard covariates describing a patient may include symptoms, prior diagnoses, family medical history, lab tests and imaging results. Nonstandard covariates may include the patient's religious practices, genetic data, wearable device data, and financial data.\footnote{See \citet{Acosta} for further examples of nonstandard patient covariates that may be predictive, but which are not currently used by clinicians for medical evaluations.} 
\end{example}

An \emph{evaluation} of the agent, $\hat{y} \in \left[-\overline{y},\overline{y}\right]$, is described in the following section. The agent has a Lipschitz continuous utility function $u: \left[-\overline{y},\overline{y}\right]^2 \rightarrow \mathbb{R}$, which maps the evaluation $\hat{y}$ and the agent's true type $y$ into a payoff. 

\begin{example}[Higher Evaluations are Better] \label{ex:PreferHigher} The agent's payoff is
\[u(\hat{y},y) =\phi(\hat{y})\]
for some increasing $\phi$. 
This corresponds, for example, to an agent receiving a desired outcome (e.g., a loan or a promotion) with probability increasing in the evaluation.
\end{example}

\begin{example}[More Accurate Evaluations are Better] \label{ex:PreferAccurate} The agent's payoff is
\[u(\hat{y},y) = -(\hat{y} - y)^2.\]
This corresponds to harms that are decreasing in the accuracy of the evaluation, e.g., medical prediction problems where more accurate evaluations are desired.     
\end{example}

\subsection{Evaluation of the agent} \label{sec:Evaluators}

There are two evaluators: a black box evaluator, henceforth Black Box (it), and a human evaluator, henceforth Human (she). Both form evaluations as an expectation of the agent's type $y$ given observed covariates, so we will introduce notation for these conditional expectations. For any covariate vector $\bold{x}_n=(x_1,\dots,x_n)$ and subset of nonstandard covariates $A \subseteq \mathcal{N}$, let
\begin{equation} \label{notation:C}
C_A(\bold{x}_n) = \{\tilde{x} \in \{0,1\}^n : \tilde{x}_i = x_i \quad \forall i \in \mathcal{S} \cup A\}
\end{equation}
be the set of all covariate vectors that agree with $\bold{x}_n$ on the covariates with indices in $\mathcal{S} \cup A$. Further define
\begin{equation}\label{eq:CondExp}\hat{y}^{f}_{\bold{x}_n}(A)= \frac{1}{\vert C_A(\bold{x}_n) \vert } \sum_{ x\in C_A(\bold{x}_n)} f(x)
\end{equation}
to be the conditional expectation of the agent's type given their standard covariates and their nonstandard covariates with indices in $A$. We use
\[U^f_{\bold{x}_n}(A) = u\left(\hat{y}^{f}_{\bold{x}_n}(A),y\right)\]
to denote the agent's payoff given this evaluation.

Both the human and black box evaluation take the form (\ref{eq:CondExp}), but the observed sets of nonstandard covariates $A$ are different across the evaluators.   Black Box observes the nonstandard covariates in the set $B=\{s+1,\dots, s+b_n\}$ where $b_n = \lfloor \alpha_b \cdot n  \rfloor$.\footnote{One can instead assume that these nonstandard covariates are selected uniformly at random. This will not affect the results of this paper.} 
Importantly, this set is held fixed across agents. So an individual with covariate vector $\bold{x}_n$ receives the evaluation $\hat{y}_{\bold{x}_n}^{f}(B)$ and payoff $U_{\bold{x}_n}^{f}(B)$ when evaluated by the Black Box.\footnote{It is not important for our results that $B$ is common across individuals; what we require is that any randomness in $B$ is independent of the agent's covariates and type. For example, if the set $B$ were drawn uniformly at random for each agent, our results would hold.}

Human differs from Black Box in two ways. First, Human has a capacity of $h_n = \lfloor \alpha_h \cdot n \rfloor$ nonstandard covariates per agent, where $\alpha_h < \alpha_b$ (i.e., Human cannot process as many inputs as Black Box).
Second, Human does not pre-specify which nonstandard covariates to observe, but rather learns these through conversation, and thus potentially observes different nonstandard covariates for each agent. For example, a doctor (evaluator) may pose different questions to different patients (agents) depending on their answers to previous questions. Or a job candidate (agent) might choose to disclose to an interviewer (evaluator)  certain nonstandard covariates that are favorable to him.  

Rather than modeling the complex process of a conversation, we study the quantity 
\begin{equation} \label{eq:UpperBound}
\max_{H \subseteq \mathcal{N}, \vert H \vert \leq \alpha_h\cdot n}  U_{\bold{x}_n}^f(H)
\end{equation}
which is the agent's payoff when the posterior expectation about his type is based on those $\alpha_h\cdot n$ or fewer covariates that are best for him. 

We can interpret this quantity as an upper bound for the agent's payoffs under certain assumptions. First, if the evaluator selects which covariates to observe, then (\ref{eq:UpperBound}) is an upper bound on the agent's possible payoffs across all possible evaluator selection rules. Second, if  covariates are disclosed by the agent, but the evaluator updates to the disclosed covariates as if they had been chosen exogenously, then again (\ref{eq:UpperBound}) represents an upper bound on the agent's possible payoffs.\footnote{\citet{JinLucaMartin} and \citet{Peregoetal} report that the beliefs of experimental subjects falls somewhere in between this naive benchmark and equilibrium beliefs, since subjects do not completely account for the strategic nature of disclosure.}

If however the covariates are disclosed by the agent in a disclosure game, and the evaluator accounts for the strategic nature of this disclosure, then whether (\ref{eq:UpperBound}) represents an upper bound will depend on what we assume that the agent knows at the time of disclosure.\footnote{This is roughly because the agent can potentially ``sneak in'' information about the other covariates via the covariates that are revealed.} We show in Section \ref{sec:Disclosure} that if the agent knows his entire covariate vector, then (\ref{eq:UpperBound}) need not upper bound every agent's payoffs. Nevertheless, we present a different quantity that does upper bound the maximum payoff that any agent can obtain in this disclosure game, and show that our main results extend when we replace (\ref{eq:UpperBound}) with this quantity. To streamline the exposition we focus on the prior two interpretations (in which the human evaluator either selects the covariates herself or updates to the agent's disclosures naively), and discuss disclosure games in Section \ref{sec:Disclosure}.

\subsection{Value of context}

A key input towards understanding the comparison between Human and Black Box is quantifying the extent to which individualized context improves the agent's payoffs.

\begin{definition}[Value of Context]
   The \emph{value of context} for an agent with covariate vector $\bold{x}_n$ and type $y=f(\bold{x}_n)$ is
\begin{equation*}
v(f,\bold{x}_n) = \max_{H \subseteq \mathcal{N}, \vert H \vert \leq \alpha_h n} U_{\bold{x}_n}^{f}(H) - U_{\bold{x}_n}^{f}(\varnothing)
\end{equation*}
i.e., the best possible improvement in the agent's utility when the evaluator additionally observes up to $\alpha_h\cdot n$ covariates for the agent.
\end{definition}

In general, the value of context depends on the type function $f$ as well as on the agent's own covariate vector $\bold{x}_n$.\footnote{The value of context given a specific function $f$ is spiritually related to   the communication complexity of $f$ \citep{Kushilevitz_Nisan_1996}.}

\begin{example}[High Value of Context] \label{ex:HighVoC}
    Let $u(\hat{y},y)=\hat{y}$, i.e., the agent's payoff is the evaluation. Suppose
$x_1$ is a standard covariate (observed no matter what), while $x_2, \dots, x_{100}$ are nonstandard covariates. The type $y$ is related to these covariates via the type function 
\[y = f(x_1, \dots, x_{100}) = \left\{ 
\begin{array}{cc}
c & \mbox{ if } x_1 = x_2\\
-c & \mbox{ if } x_1 \neq x_2
\end{array}\right.\]
For an agent who can reveal (up to) one covariate and whose covariate vector is $(1,1,\dots,1)$, the value of context is $c$, since revealing $x_2=1$ moves the expectation of his type from 0 to $c$. This example corresponds to settings in which some  nonstandard covariate substantially moderates the interpretation of a standard covariate. For such type functions $f$, it is important for the evaluator to observe the right nonstandard covariates, and so the value of context can be large.
\end{example}

\begin{example}[Low Value of Context]
Suppose the type function in the previous example is instead  $y=f(x_1)=x_1$ (leaving all other details of the example unchanged). Then the value of context is 0 for every agent. In this example, nonstandard covariates are irrelevant for predicting the type, so there is no value to the evaluator discovering the ``right'' covariates.
\end{example}

In what follows, we give the agent uncertainty about $f$ and characterize the agent's expected value of context and expected payoffs, integrating over the agent's belief about $f$.\footnote{If we interpret the covariates in our model as signals about the type, then the function relating covariates to type corresponds to the signal structure.
} 

We do this for two reasons.  First, in many applications it is not realistic to suppose that the agent knows $f$. For example, a patient who anticipates that a diagnosis will be based on an image scan of his kidney may recognize that there are  properties of the image that are indicative of whether he has the condition or not, but likely does not know what the relevant properties are, or how they determine the diagnosis.\footnote{In the case of a job interview, the function $f$ may reflect particular subjective preferences of the firm, which are initially unknown to the agent.}

Second, the case with uncertainty about $f$ turns out to yield a more elegant analysis than the one in which $f$ is known. That is, although the value of context for specific functions $f$ depends on details of that function and on the agent's own covariate vector, there is a large class of prior beliefs (described in the following section) for which it is possible to draw strong detail-free conclusions about the expected value of context.

\subsection{Model Uncertainty}

We impose two assumptions on the agent's prior belief about $f$. Together, these assumptions deliver a setting in which many different structural relationships between the covariates and the  type are possible (including both ones where the value of context is high and low), but ex-ante those relationships are not known.  

The first assumption says that while the agent may know how standard covariates impact the type, he  has no ex-ante knowledge about the roles of the nonstandard covariates. 

\begin{assumption}[Exchangeability] For every realization of the standard covariates $(x_1, \dots, x_s)$, the sequence
\begin{equation} \label{eq:Exchangeable}
(Y^n_1, \dots, Y^n_{2^{n-s}}) \equiv (f(x_1,\dots,x_s,x_{s+1}, \dots, x_n) : (x_{s+1},\dots,x_n) \in \{0,1\}^{n-s})
\end{equation}
is finitely exchangeable. \label{assp:Exchangeability}
\end{assumption}

The set $\{(f(x_1,\dots,x_s,x_{s+1}, \dots, x_n) : (x_{s+1},\dots,x_n) \in \{0,1\}^{n-s}\}$ ranges over all covariate vectors that share the standard covariate values $(x_1,\dots,x_s)$. Assumption \ref{assp:Exchangeability} says that the joint distribution of these types is ex-ante invariant to permutations of the labels and values of the nonstandard covariates. An agent whose prior satisfies Assumption \ref{assp:Exchangeability} is thus agnostic about how the nonstandard covariates impact the type.  

 While our assumption of no prior knowledge about the role of nonstandard covariates is strong, it is consistent with our interpretation of the nonstandard covariates as those covariates for which there is little historical data about correlations. For example, if it were known  that a higher GPA positively correlates with  on-the-job performance, but not how a large number of social media followers predicts performance, then we would think of GPA as a standard covariate and the number of social media followers as a nonstandard covariate.

Assumption \ref{assp:Exchangeability} does not restrict how the agent's prior varies with $n$, the number of covariates. In a model in which $x_1,x_2,\dots$ were drawn i.i.d.\ from a type-dependent distribution $F_y$, the total quantity of information about $y$ would increase in the number of covariates, and the evaluator's uncertainty about $y$ would vanish as $n$ grew large. This does not seem descriptive of real applications: credit scoring algorithms and healthcare algorithms use millions of covariates, but there remains substantial residual uncertainty about the agent's type. We take the opposite extreme in which the predictability of the agent's type is a primitive of the setting, which is held constant for all $n$. In our model, $n$ is not a measure of the total quantity of information, but instead moderates the richness of the informational environment and the potential complexity of the mapping $f$. Loosely speaking, as $n$ grows large,  the agent has a more extensive set of words to describe a fixed unknown $y$.\footnote{As $n$ grows large, the smallest possible informational size of each covariate (in the sense of \citet{McLeanPostlewaite}) vanishes. But we do not require each covariate to be equally informationally relevant in the realized function. So, for example, $f(x_1,\dots,x_n)=x_1$ can be in the support of the agent's beliefs for $n$ arbitrarily large (see Example \ref{ex:Binary}).}

\begin{assumption}[Constant Unpredictability of $Y$] Fix any realization of the standard covariates $(x_1, \dots, x_s)$, and  let 
$(Y^n_1, \dots, Y_{2^{n-s}}^n)$ be defined as in (\ref{eq:Exchangeable}) for each $n \in \mathbb{Z}_+$. Then for every pair $n'>\tilde{n}$, the sequence $(Y^{\tilde{n}}_1, \dots, Y^{\tilde{n}}_{2^{\tilde{n}-s}})$ and the truncated sequence $(Y^{n'}_1, \dots, Y^{n'}_{2^{\tilde{n}-s}})$ are identically distributed.  \label{assp:ConstantVarY} 
\end{assumption}

The statement of Assumption  \ref{assp:ConstantVarY}  formally depends on the ordering of types within the vector $(Y^{n'}_1, \dots, Y^{n'}_{2^{n'}})$, since this determines which types appear in the truncated sequence $(Y^{n'}_1, \dots, Y^{n'}_{2^n})$. But if we further impose Assumption \ref{assp:Exchangeability} (and we will always impose these assumptions jointly), then the ordering of types is irrelevant: That is, when Assumption \ref{assp:ConstantVarY} holds for one such ordering, it will hold for all orderings. 

\bigskip

It is important to note that in our model, Assumptions \ref{assp:Exchangeability} and \ref{assp:ConstantVarY} are placed \emph{ex-ante} on the agent's prior, and not \emph{ex-post} on the realized function $f$. For example, the function $f(x_1, \dots, x_n)=x_1$, which says that the only covariate that matters is $x_1$, is strongly asymmetric ($x_1$ is differentiated from the other covariates) and also  features a single ``large'' covariate (the realization of $x_1$ completely determines $y$). Our assumptions do not rule out the possibility of this function. Rather, they require that if this function is considered possible, then certain other functions are as well.\footnote{Assumption \ref{assp:Exchangeability} requires that for every permutation $\pi: \{0,1\}^n \rightarrow \{0,1\}^n$, the function $g_\pi$ satisfying $g_\pi(x_1,\dots,x_n)=f(\pi(x_1, \dots, x_n))$ is also in the support of the agent's beliefs.}

Simple examples of priors satisfying these two assumptions are given below.

\begin{example} \label{ex:Binary} Let $y\in \{0,1\}$, in which case the space of possible functions $f: \mathcal{X} \rightarrow \mathcal{Y}$ can be identified with $\{0,1\}^{2^n}$. Suppose that for each $n$, the agent has a uniform prior on the set of all functions $\{0,1\}^{2^n}$. Then Assumptions \ref{assp:Exchangeability} and \ref{assp:ConstantVarY} are satisfied.\end{example}

\begin{example} Suppose there is a distribution $F$ on $[-\overline{y},\overline{y}]$ such that for each $n$, 
\[\left(f(x_1,\dots,x_s,x_{s+1}, \dots, x_n) : (x_{s+1},\dots,x_n) \in \{0,1\}^{n-s}\right) \sim_{i.i.d.} F.\]
Then Assumptions \ref{assp:Exchangeability} and \ref{assp:ConstantVarY} are satisfied.
\end{example}

\bigskip

Priors that violate these assumptions include the following. 

\begin{example}[Only One Covariate is Relevant] \label{ex:OnlyOne}
The type is equal to the value of the nonstandard covariate $x_I$, where the index $I$ is drawn uniformly at random from $\mathcal{N}$. Then Assumption \ref{assp:Exchangeability}  fails.\footnote{Suppose $n=2$, and both covariates are nonstandard. Then under the agent's prior, $f \in \left\{\hat{f},\tilde{f}\right\}$ where $\hat{f}(1,1)=\hat{f}(1,0)=1$ while $\hat{f}(0,1)=\hat{f}(0,0)=0$, and $\tilde{f}(1,1)=\tilde{f}(0,1)=1$ while $\tilde{f}(1,0)=\tilde{f}(0,0)=0$. So the agent knows with certainty that $f(1,1)=1$ but $f(0,0)=0$, in violation of Assumption \ref{assp:Exchangeability}. This example is consistent with exchangeability in the labels of the nonstandard covariates, but not with exchangeability in their realizations.}
\end{example}

\begin{example}[Higher Values are Better] The value of $f(\bold{x}_n)$ is (independently) drawn from a uniform distribution on $[1,2]$ if $x_{s+1}=1$, and (independently) drawn from a uniform distribution on $[0,1]$ if $x_{s+1}=0$. Then Assumption \ref{assp:Exchangeability}  fails.
\end{example}

We view these two assumptions as useful conceptual benchmarks, but neither is necessary for our subsequent results. In Section \ref{sec:ExtendBreak}, we explore how far our main results generalize under different relaxations of Assumptions \ref{assp:Exchangeability} and \ref{assp:ConstantVarY}.
 
\subsection{Expected Value of Context}

We now define the expected value of context from the point of view of an agent who knows his covariate vector $\bold{x}_n$ but does not know the function $f$ (and hence also does not know his type $y=f(\bold{x}_n)$). As we show in Section \ref{sec:Max}, the assumption that the agent knows $\bold{x}_n$ is immaterial for the results. 

\begin{definition}[Expected Value of Context] For every $n \in \mathbb{Z}_+$ and covariate vector $\bold{x}_n \in \{0,1\}^n$, the \emph{expected value of context} is
\[
V(n,\bold{x}_n) = \mathbb{E}\left[v(f, \bold{x}_n)\right].
\] \label{def:ExpectedVoC}
\end{definition}

\noindent This quantity tells us the extent to which  context improves the agent's payoffs in expectation.

We similarly compare evaluators based on the expected payoff that the agent receives.

\begin{definition} \label{def:CompareEvaluators} Consider any agent with covariate vector $\bold{x}_n$. If
\begin{equation} \label{def:PreferA}
\mathbb{E}\left[\max_{H \subseteq \mathcal{N}, \vert H \vert \leq \alpha_h\cdot n} U_{\bold{x}_n}^{f}(H)\right] < \mathbb{E}\left[ U_{\bold{x}_n}^{f}(B)\right]
\end{equation}
then say that the agent \emph{prefers the black box evaluator}. And if
\begin{equation} \label{def:PreferH}
\mathbb{E}\left[\min_{H \subseteq \mathcal{N}, \vert H \vert \leq \alpha_h\cdot n} U_{\bold{x}_n}^{f}(H)\right] > \mathbb{E}\left[ U_{\bold{x}_n}^{f}(B)\right]
\end{equation}
then say that the agent \emph{prefers the human evaluator}. 
\end{definition}

These definitions correspond to a thought experiment in which (for example) a patient has a choice between being seen by a doctor or assessed by an algorithm. If the patient chooses the algorithm, his standard covariates and $\alpha_b \cdot n$ arbitrarily chosen 
nonstandard covariates will be sent to the algorithm. If the patient chooses the doctor, he will engage in a conversation with the doctor, where his standard covariates and $\alpha_h \cdot n$ selected nonstandard covariates will be revealed.  Which should the patient choose?

The first part of Definition \ref{def:CompareEvaluators} compares the agent's expected payoff under black box evaluation with the \emph{best-case} expected payoff under human evaluation, namely when the human evaluator observes those (up to) $\alpha_h\cdot n$ covariates that maximize the agent's payoffs. If the agent's expected payoff is nevertheless higher under black box evaluation even after biasing the agent towards the human in this way, we say that the agent \emph{prefers to be evaluated by the black box}. The second part of the definition compares the agent's expected payoff under black box evaluation with the \emph{worst-case} expected payoff under human evaluation, namely when the human evaluator observes those (up to) $\alpha_h\cdot n$ covariates that minimize the agent's payoffs. If the agent's expected payoff is lower under black box evaluation even after biasing the agent against the human in this way, then we say that the agent \emph{prefers to be evaluated by the human}.\footnote{In Section \ref{sec:Disclosure} we further discuss the extent to which these interpretations are valid when the evaluator also updates her beliefs to the selection of covariates.}

These are clearly very conservative criteria for what it means to prefer the human or the black box. In practice, we would expect the set of revealed covariates $H$ to be intermediate to the two cases considered in Definition \ref{def:CompareEvaluators}, i.e., that $H$ neither maximizes nor minimizes the agent's payoffs.\footnote{\citet{Angelovaetal} provide evidence that some judges condition on irrelevant defendant covariates when predicting misconduct rates.} But if we can conclude either that the agent prefers the black box evaluator or the human evaluator according to Definition \ref{def:CompareEvaluators}, then the same conclusion would hold for these more realistic models of $H$.

\section{Main Results}

 Section \ref{sec:ValueofExplanation} characterizes the expected value of context in a simple example. Section \ref{sec:ExpectedValue} presents our first main result, which says that the expected value of context vanishes to zero as the number of covariates grows large. Section \ref{sec:HvBB} compares human and black box evaluators.

\subsection{Example} \label{sec:ValueofExplanation}

Suppose there are two covariates, $x_1$ and $x_2$, both nonstandard. For each covariate vector $\bold{x} \in \{0,1\}^2$, define the random variable $Y_{\bold{x}} = f(\bold{x})$, where the randomness is in the realization of $f$. 

\begin{table}[H]
    \[\begin{array}{ccc}
X_1 & X_2 & Y_{\bold{x}}\\
0 & 0 & Y_{00} \\
0 & 1 & Y_{01} \\
1 & 0 & Y_{10} \\
1 & 1 & Y_{11}
\end{array}\]
\caption{The four possible covariate vectors and their associated types.}
    \label{tab:BinaryExample}
\end{table}

The agent has utility function $u(\hat{y},y)=\hat{y}$ and covariate vector $(1,1)$. Suppose Human observes up to one nonstandard covariate; then, there are three possibilities for what the evaluator observes.  If Human observes $x_1=1$,  her evaluation is 
\[Z_1 \equiv \frac{Y_{10}+Y_{11}}{2}.\]
If Human observes $x_2=1$,  her evaluation is 
\[Z_2 \equiv \frac{Y_{01}+Y_{11}}{2}.\] And if Human observes no nonstandard covariates, then her evaluation remains the unconditional average \[Z_{\varnothing} \equiv \frac{Y_{00}+Y_{01}+Y_{10}+Y_{11}}{4}.\]
So the expected value of context for this agent is
\begin{equation} \label{eq:VoC}
\mathbb{E}\left[\max\left\{ Z_\varnothing, Z_1,Z_2 \right\}  - Z_\varnothing\right].
\end{equation}

\bigskip

Suppose $n$ grows large with up to $h_n=\lfloor\frac{n}{2}\rfloor$ covariates observed. There are two opposing forces affecting the value of context. First, when $n$ is larger there are more distinct sets of covariates that can be revealed to Human, and hence the max in (\ref{eq:VoC}) is taken over a larger number of posterior expectations. This increases the value of context. On the other hand,  each $Z_k$ is a sample average, and the number of elements in this sample average also grows in $n$.\footnote{For example, observing $X_1=1$  with $n=2$  gives the evaluator a posterior expectation of $(Y_{10}+Y_{11})/2$, while the same observation gives the evaluator a posterior expectation of $(Y_{100}+Y_{101}+Y_{110}+Y_{111})/2$ if $n=3$.} By the law of large numbers, each $Z_k$  thus concentrates on its expectation (which is common across $k$) as $n$ grows large, so the difference between any $Z_k$ and $Z_{k'}$ grows small. What we have to determine is whether the growth rate in the number of subsets of nonstandard covariates (of size $\leq h_n$) is sufficiently large such that the maximum difference in evaluations across these sets is nevertheless asymptotically bounded away from zero. The answer turns out to be no.

\subsection{The Expected Value of Context} \label{sec:ExpectedValue}

Our main result says that for every agent, the expected value of context (as defined in Definition \ref{def:ExpectedVoC}) vanishes as $n$ grows large.

\begin{theorem} \label{thm:Main} Suppose Assumptions \ref{assp:Exchangeability} and \ref{assp:ConstantVarY} hold. Then for every covariate vector $\bold{x} \in \{0,1\}^\infty$, the expected value of context vanishes to zero as $n$ grows large, i.e.,
	\[\lim_{n\rightarrow\infty}V(n,\bold{x}_n)=0.\]
\end{theorem}

Thus although the value of context may be substantial for certain type functions (such as in Example \ref{ex:HighVoC}), it does not matter on average across these functions when the agent's prior satisfies Assumptions \ref{assp:Exchangeability} and \ref{assp:ConstantVarY}. This also implies that for sufficiently large $n$, the provision of context does not ``typically" matter; that is, the probability that the agent  gains substantially from targeted information acquisition is small.

The core of the proof of Theorem \ref{thm:Main} is an argument that the extent to which context can change the evaluator's posterior expectation vanishes in the number of covariates. We outline that argument here. For each $n$, there are $K_n=\sum_{j=0}^{\lfloor \alpha_h  n \rfloor}\binom{n-s}{j}$  sets of $\alpha_h n$ (or fewer) nonstandard covariates that can be disclosed. We can enumerate and index these sets to $k=1, \dots, K_n$. Each set $k$ induces a posterior expectation $Z_k$ which is a sample average of random variables $Y_x \equiv f(x)$. The expected value of context (for this  utility function) is \[\E\left[\max_{1\le k\le K_n} Z_k \right]-\E[Z_\varnothing]\]
where $Z_\varnothing$ is Human's posterior expectation given observation of standard covariates only. Normalizing $E[Z_\varnothing]=0$, it remains to study properties of the first-order statistic  $\max_{1\le k\le K_n}Z_k$.

There are two  challenges to analyzing this quantity. First, the correlation structure of $Z_1, \dots, Z_{K_n}$ can be complex: The variables $Z_k$ are neither independent (because the same posterior expectation $Y_x$ can appear as an element in different sample averages $Z_k,Z_{k'}$) nor  identically distributed (because the sample averages are of different sizes depending on how many nonstandard covariates are revealed). The second challenge is that the length of the sequence $(Z_1,\dots, Z_{K_n})$ grows exponentially in $n$. Thus even though each term within the maximum eventually  converges to a normally  distributed random variable (with shrinking variance), the errors of each term may in principle accumulate in a way that the maximum grows large.

Our approach is to first construct new i.i.d.\ variables $Z^{iid}_k$, with the property that
\begin{equation} 
\mathbb{E}\left[\max\{Z_1, \dots, Z_{K_n}\}\right] \leq  \mathbb{E}\left[\max\{Z^{iid}_1, \dots, Z^{iid}_{K_n}\}\right]
\end{equation}
Applying a result from \cite{chernozhukov2013gaussian}, we show that
$\max_{1\le k\le K_n} Z_k^{iid}$ (properly normalized) converges to $\max_{1\le k\le K_n} Z_k^{Normal}$ in distribution, where (due to properties of our problem) $Z_k^{Normal}\sim_{iid} \mathcal{N}\left(0,\frac{1}{2^{n(1-\alpha_h)-s}}\right)$. Having reduced the analysis to studying the expected maximum of i.i.d.\, Gaussian variables, classic bounds apply to show that this quantity is no more than  \begin{equation} \label{eq:Bound} 
\frac{1}{2^{n(1-\alpha_h)-s}}\sqrt{\log(K_n)}.
\end{equation}
This display quantifies the importance of each of the two forces discussed in Section \ref{sec:ValueofExplanation}. First, as $n$ grows larger, the number of  posterior expectations $K_n = \sum_{j=0}^{\lfloor \alpha_h n\rfloor}\binom{n-s}{j} \leq 2^{n-s}$ grows exponentially in $n$, increasing the expected value of context. But second, as $n$ grows larger, each $Z_k$ concentrates on its expectation, where its variance, $\frac{1}{2^{n(1-\alpha_h)-s}}$, decreases exponentially in $n$. What the bound in display (\ref{eq:Bound}) tells us is that the exponential growth in the number of variables is eventually dominated by the exponential reduction in the variance of each variable,  yielding the result.

\bigskip

This proof sketch also clarifies the role of Assumption \ref{assp:Exchangeability}. As we show in Section \ref{sec:GeneralLemma}, the statement of the theorem extends so long as the evaluator's posterior expectation $Z_k$ concentrates on its expectation sufficiently quickly as $n$ grows large. Roughly speaking, this means that the informativeness of any specific set of covariates is decreasing in the total number of covariates. Thus the precise symmetry imposed by Assumption \ref{assp:Exchangeability} is not critical for Theorem \ref{thm:Main} to hold.

On the other hand, the conclusion of Theorem \ref{thm:Main} can fail  if the agent has substantial prior knowledge about how $y$ is related to the covariates. 

\begin{example} 
\label{ex:Structural} Let $s=0$, so that there are no standard covariates. Suppose that for each $n$,
\[y=f(x_1, \dots, x_n)= \frac{1}{n} \sum_{i=1}^n x_i \cdot U\]
where $U$ is a uniform random variable on $[0,1]$. This model violates Assumption \ref{assp:Exchangeability}, since it is known that higher realizations of the agent's covariates are good news about the agent's type. The conclusion of Theorem \ref{thm:Main} also does not hold: For any $n$, the evaluator's prior expectation is $\mathbb{E}[f(\bold{x}_n)] = 1/4$. But if $\lfloor \alpha \cdot n\rfloor$ covariates are revealed to be 1, the evaluator's posterior expectation is  equal to $\frac14 + \frac14 \frac{\lfloor \alpha n \rfloor}{n}$. So the expected value of context for an agent with $\mathbf{x}_n=(1,...,1)$ is asymptotically bounded away from zero.  
\end{example}

In Section \ref{sec:ExtendBreak} we explore several relaxations of Assumptions \ref{assp:Exchangeability} and \ref{assp:ConstantVarY}. Our first relaxation of Assumption \ref{assp:Exchangeability} supposes that there is a ``low-dimensional'' set of covariates that predictive of the agent's type, while the remaining covariates are irrelevant. The second relaxation supposes that there is a subset of nonstandard covariates whose effects are known. We also consider a relaxation of Assumption \ref{assp:ConstantVarY} where the evaluator's ability to predict $Y$ is increasing in the total number of covariates that define the type. We formalize these extensions of our main model and examine the extent to which Theorem \ref{thm:Main} extends.

\subsection{Human versus Black Box} \label{sec:HvBB}

We now turn to the question of when the agent should prefer the human evaluator and when the agent should prefer the black box evaluator.

\begin{assumption} \label{assp:Phi} The agent's expected utility can be written as $\mathbb{E}[\phi(\hat{y})]$ for some  twice continuously differentiable function $\phi$.\footnote{Restricting to utility functions that depend on a posterior mean is a common assumption in the literature on information design,  see e.g., \citet{KamenicaGentzkow},  \citet{Frankel} and \citet{DworczakMartini}.} Moreover, there exists $C <\infty $ such that
\[\frac{\sup_{\hat y\in[-\overline{y},\overline{y}]}|\phi'(\hat y)|}{\inf_{\hat y\in[-\overline{y},\overline{y}]}|\phi''(\hat y)|}< \frac{C}{2}.\]
\end{assumption}

Roughly speaking, the numerator describes the sensitivity of the function $\phi$ to the evaluation, and the denominator describes the curvature of the function $\phi$. The assumption thus says that the curvature of the function must be sufficiently large relative to its slope. While there is no formal relationship, the LHS is evocative of the coefficient of absolute risk aversion of the function $\phi$.\footnote{Recall that the coefficient of absolute risk aversion of the function $\phi$ is $-\frac{\phi'(\hat{y})}{\phi''(\hat{y}))}$.}

\begin{theorem} \label{prop:PreferHA}
	Suppose Assumptions \ref{assp:Exchangeability}-\ref{assp:Phi} hold, and let
 \begin{equation}
          \label{eq:N}
         N = \min\left\{n\in \mathbb{R}_+ :  (\alpha_b-\alpha_h)n -\frac{1}{2}\log_2(n)-1>\log_2(C)\right\}.
          \end{equation}
          Then:
 \begin{itemize}

     \item[(a)] If $\phi$ is strictly convex, the agent prefers the black box evaluator for all $n\geq N$.
          \item[(b)] If $\phi$ is strictly concave, the agent prefers the human evaluator for all $n\geq N$.
 \end{itemize}
\end{theorem}

The comparisons in this theorem apply  for reasonably small $N$. For example, let $C=100$, in which case the restriction in Assumption \ref{assp:Phi} is quite weak. Figure \ref{fig:N} fixes $\alpha_h=0.1$ and plots $N$ for different values of $\alpha_b$. If (say), the human evaluator observes 10\% of covariates while Black Box observes 90\%, then the comparisons in Theorem \ref{prop:PreferHA} hold for all $n \geq 14$.

\begin{figure}[h]
\begin{center}
\includegraphics[scale=0.4]{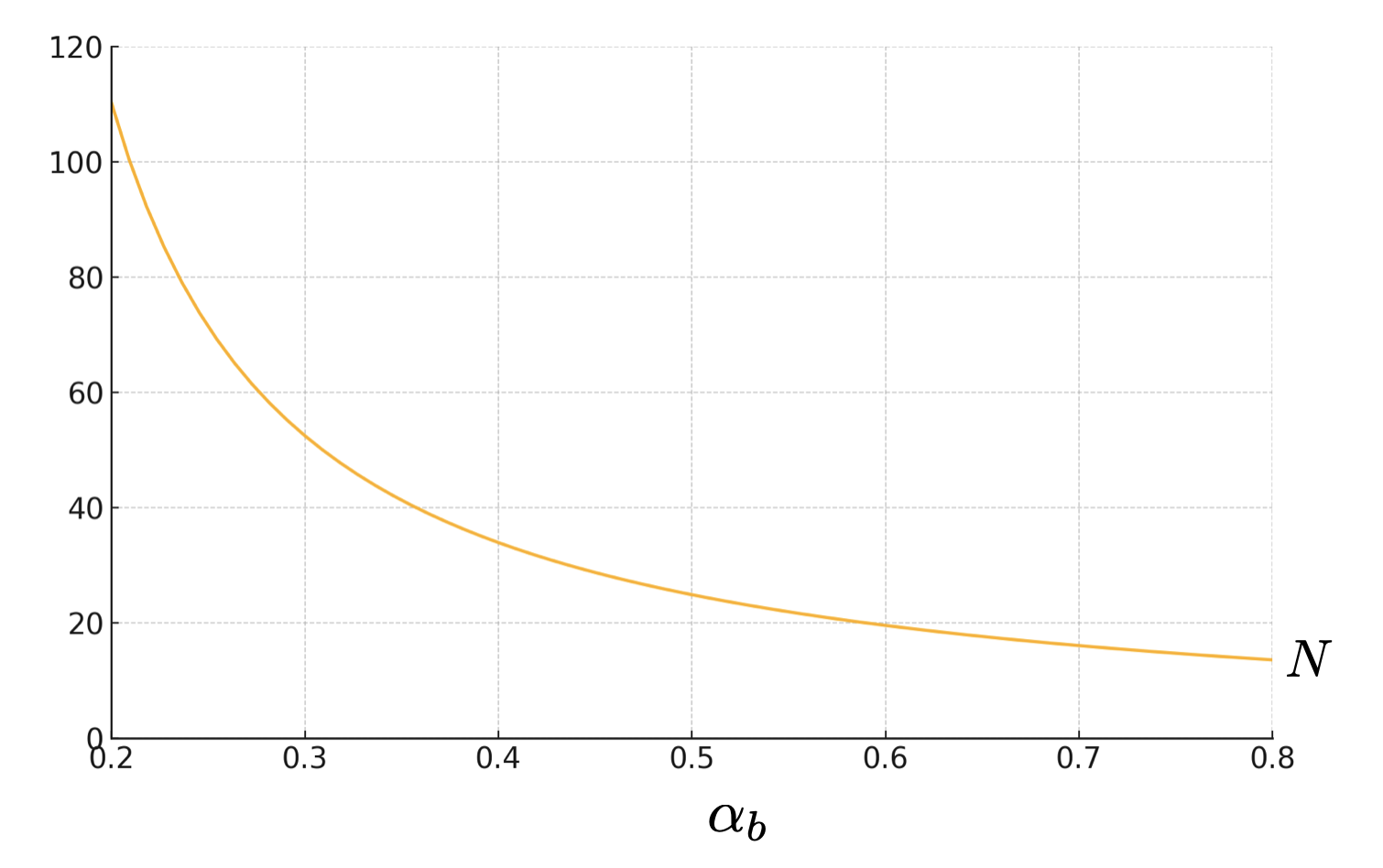}
    \end{center}
    \caption{Let $C=100$ and $\alpha_h=0.1$. Then the comparisons in Theorem \ref{prop:PreferHA} hold for all $n \geq N$ as depicted here.} \label{fig:N}
\end{figure}

The case of convex $\phi$ (Part (a)) corresponds to a preference for more accurate evaluations.\footnote{Consider any two sets of covariates $A \subset A'$ and let $\hat{y}_A$, $\hat{y}_{A'}$ be the corresponding posterior expectations. The distribution of $\hat{y}_{A'}$ (i.e., the posterior expectation that conditions on more information) is a mean-preserving spread of the distribution of $\hat{y}_A$. When $\phi$ is convex, the former leads to a higher expected utility. Such an agent ``prefers more accurate evaluations'' in the sense that giving the evaluator better information (in the standard Blackwell sense) leads to an improvement in the agent's expected utility.} Such an agent prefers for the evaluation to be based on more information (advantaging Black Box), but also prefers for the evaluation to be based on more relevant covariates (advantaging Human). We show that what eventually dominates is how many covariates the evaluators observe, not how they are selected; for an agent who prefers accuracy, this favors the Black Box.

Part (b) of Theorem \ref{prop:PreferHA} says that if instead $\phi$ is concave, then the agent should eventually prefer the human evaluator. We conclude this section with example decision problems that induce utility functions satisfying the conditions of either part of the theorem.

\begin{example} Suppose the agent receives a dollar wage equal to the evaluation, and is risk averse in money. Then his utility function is $u(\hat{y},y) = \phi(\hat{y})$ for some increasing and concave $\phi$, and Part (a) of Theorem \ref{prop:PreferHA} says that the agent prefers to be evaluated by the human.
\end{example}

\begin{example}
Suppose the agent's type is $y\in \{0,1\}$, and the evaluator chooses an action $a$ based on the observed covariates. The evaluator and agent share the utility function $-\mathbb{E}[(a-y)^2]$. The evaluator's optimal action is $a=\hat{y}$, and the agent's expected payoff given this action is
\[\mathbb{E}\left[-(\hat{y}-y)^2\right] = \mathbb{E}\left[-\left(\hat{y}(1-\hat{y})^2 + (1-\hat{y})(\hat{y})^2\right)\right] = \mathbb{E}\left[\phi(\hat{y})\right]\] where $\phi(\hat{y}) = \hat{y}^2-\hat{y}$ is convex. So Part (a) of Theorem \ref{prop:PreferHA} says that the agent eventually prefers evaluation by the black box evaluator. Although the conditions of Theorem \ref{prop:PreferHA} are no longer met when $y$ is not binary, we show in Appendix \ref{app:sqLoss} that the conclusion of Part (a) of Theorem \ref{prop:PreferHA} generalizes for arbitrary $y$ given the mean-squared error payoff function described in this example. 
\end{example}

\section{Extensions} \label{sec:Extensions}

We now show that we are able to strengthen our main results (Theorems \ref{thm:Main} and \ref{prop:PreferHA}) in the following ways. In Section \ref{sec:Max}, we show that not only does the expected value of context vanish for each individual agent, but in fact the expected maximum value of context across agents also vanishes. That is, in expectation the most that context can benefit \emph{any} agent in the population is small. From this, it is immediate that our main results also extend in a generalization of our model in which the agent has uncertainty over his covariate vector. In Section \ref{sec:Disclosure}, we show that our main results extend when the agent and evaluator interact in a disclosure game, wherein the evaluator updates his beliefs to the agent's strategic choice of what to disclose.

\subsection{Max value of context across agents} \label{sec:Max}

So far we have studied the the expected value of context for a single agent. If we instead ask whether the firm  should use human or algorithmic evaluation---for example, whether a hospital should automate diagnoses or rely on doctor evaluations---other statistics may also be relevant. For example, it may matter whether the value of context is large for any agent in the population (e.g., because a lawsuit regarding algorithmic error may be brought on the basis of harm to any specific individual \citep{Jha}). We thus study the expected maximum value of context, as defined below.

\begin{definition} For any $n \in \mathbb{Z}_+$, the \emph{expected maximum value of context} is
\[
V^{\text{max}}(n) = \mathbb{E}\left[\max_{\bold{x}_n \in \{0,1\}^n} v(f, \bold{x}_n)\right].
\]
\end{definition}

The following corollary says that this quantity also vanishes as $n$ grows large.

\begin{corollary} \label{corr:Max}
Suppose Suppose Assumptions \ref{assp:Exchangeability} and \ref{assp:ConstantVarY} hold.
Then the expected maximum value of context vanishes to zero as $n$ grows large, i.e.,
	$\lim_{n\rightarrow\infty}\overline{V}^{\text{max}}(n)=0.$
\end{corollary}

\noindent Thus, the expected value of context vanishes uniformly across agents in the population. This result immediately implies that Theorems \ref{thm:Main} and \ref{prop:PreferHA} extend in any generalization of our model in which the agent has uncertainty not only over $f$ but also over his own covariate vector $\bold{x}_n$.

\subsection{Strategic Disclosure} \label{sec:Disclosure}

So far we've remained agnostic as to whether the agent or evaluator chooses which nonstandard covariates are revealed, assuming  that in either case the evaluator  updates as if the covariates were revealed exogenously. We now consider a more traditional disclosure game, in which the agent chooses which nonstandard covariates are revealed, and the human evaluator updates her beliefs about the agent's type in part based on which covariates are chosen.

 For any fixed function $f$,  call the following an \emph{$f$-context disclosure game}:  There are two players, the agent and the evaluator. The function $f$ is common knowledge.\footnote{We do not interpret this assumption literally. At the other extreme where $f$ is unknown to the agent, there is no informational content in which covariates the agent chooses to reveal, as they are all symmetric from the agent's point of view.} 
The set of possible \emph{disclosures} $\mathcal{D}$ is the set of all pairs $(H, (x_i)_{i \in H})$ consisting of a set of nonstandard covariates $H \subseteq \mathcal{N}$ and values for those covariates. A disclosure $d = (H,(x'_i)_{i \in H})$ is \emph{feasible} for an agent with covariate vector $(x_1,\dots,x_n)$ if the disclosed covariate values are truthful, i.e., $x_i = x_i'$ for every $i \in H$.

The agent chooses a \emph{disclosure strategy}, which is a map
\[\sigma: \{0,1\}^n \rightarrow \mathcal{D}\] 
from covariate vectors to feasible disclosures. The agent then privately observes his covariate vector $\bold{x}_n$ and discloses $\sigma(\bold{x}_n)$. The evaluator observes this disclosure and chooses an action $\hat{y}$. That is, the evaluator's strategy is a function $\sigma_E: \mathcal{D} \rightarrow [-\overline{y},\overline{y}]$. The evaluator's payoff is $-(\hat{y}-y)^2$ and the agent's payoff is some function $u(\hat{y})$. 

In this section we focus on pure strategy Perfect Bayesian Nash equilibria (PBE) of this game, henceforth simply \emph{equilibria}. (A similar result holds for mixed strategy equilibria, which is demonstrated in the appendix.)   

\begin{definition}
Let $v^D(f,\bold{x}_n)$ denote the highest payoff that an agent with covariate vector $\bold{x}_n$ receives in any pure-strategy
equilibrium of the $f$-context disclosure game. The \emph{expected maximum value of context disclosure} is
  \[V^D(n) = \mathbb{E}\left[\max_{\bold{x}_n \in \{0,1\}^n} v^D(f,\bold{x}_n)\right].\]
\end{definition}

We show that the best payoff that an agent can receive in any pure strategy $f$-context equilibrium is bounded above by the maximum value of context across agents.

\begin{proposition} \label{prop:DisclosureBound} Suppose Assumptions \ref{assp:Exchangeability} and \ref{assp:ConstantVarY} hold. Then for all $n$,
\[V^D(n)\le V^{\text{max}}(n).\]
\end{proposition}
Thus, applying Proposition \ref{prop:DisclosureBound} and Corollary \ref{corr:Max}, our previous results extend.

\section{Relaxing our Assumptions on the Prior} \label{sec:ExtendBreak}

As shown in Example \ref{ex:Structural}, our main results can fail if the assumption of symmetric uncertainty over the role of the nonstandard covariate values (Assumption \ref{assp:Exchangeability}) is broken. We now propose two relaxations of Assumption \ref{assp:Exchangeability} and one relaxation of Assumption \ref{assp:ConstantVarY}, and explore the extent to which our main result extends.  In Section \ref{subsec:Irrelevance}, we suppose that it is known ex-ante that some $r_n$ covariates are relevant, while the remaining $n-r_n$ are not, so that even as $n$ grows to infinity, the effective number of covariates potentially grows more slowly. In Section \ref{sec:KnownEffect} we allow the agent to have prior knowledge about the role of certain nonstandard covariates. In Section \ref{sec:NonConstant}, we consider a model in which the predictability of the agent's type is increasing in the total number of covariates. Finally, Section \ref{sec:GeneralLemma} provides an abstract condition on the learning environment under which our main results hold, which requires the evaluator's uncertainty about the agent's type to grow sufficiently fast in $n$.

\subsection{Not all covariates are relevant} \label{subsec:Irrelevance}

Under Assumption \ref{assp:Exchangeability}, it cannot be known ex-ante that some covariates are irrelevant for predicting the type. The assumption thus rules out settings such as the following.

\begin{example} The evaluator is a job interviewer. Although in principle there is an infinite number of covariates that can describe a job candidate, it is understood that not all of them are actually relevant to the job candidate's ability. That is, there is some potentially large (but not exhaustive) set of covariates that contain all of the predictive content about the candidate's ability, and the remaining covariates are either irrelevant for predicting ability, or are predictive only because they correlate with other intrinsically predictive covariates. 
\end{example}

If irrelevant covariates cannot be disclosed to the evaluator, then we return to  our main model with a smaller $n$ and our previous results extend directly.  The more novel case is the one in which it is known that $n-r_n$ covariates are irrelevant, but those covariates can still be disclosed to the evaluator (for example, because it is not commonly understood that they are irrelevant).\footnote{To see the difference, consider the case in which the agent simply wants the evaluator to hold a higher posterior expectation. The irrelevant covariates create noise, and for some realizations of $f$ it may be that disclosing an irrelevant covariate leads to a higher evaluation.}

To model this, we suppose there is a sequence of sets of \emph{relevant covariates} $(R_1,R_2, \dots)$ such that each $R_n$ includes the standard covariates in $\mathcal{S}$ and is of size $s+r_n$, where $r_n = \lfloor \alpha_r \cdot n \rfloor$ is the (known) number of relevant nonstandard covariates. We will moreover without loss index these sets so that the relevant covariates are $x_{s+1},x_{s+2},\dots,x_{s+r_n}$. The irrelevance of the remaining covariates is reflected in the following assumption, which says that, holding fixed the values of the relevant covariates, the values of the irrelevant covariates do not change the type.

\begin{assumption}[Irrelevance] There is a function $g(x_1, \dots,x_{s+r_n})$ such that
\[f(x_1,\dots,x_n) = g(x_1, \dots,x_{s+r_n})\]
for every $(x_1,\dots, x_n) \in \{0,1\}^n$.
\end{assumption}

We then modify Assumptions \ref{assp:Exchangeability} and \ref{assp:ConstantVarY} to apply  only to the relevant covariates.

\begin{assumption}[Exchangeability] For every realization of $(x_1, \dots, x_s)$, the sequence
\begin{equation}
(\widetilde{Y}^n_1, \dots, \widetilde{Y}^n_{2^n}) \equiv (g(x_1,\dots,x_s,x_{s+1}, \dots, x_{s+r_n}) : (x_{s+1},\dots,x_{s+r_n}) \in \{0,1\}^{r_n-s})
\end{equation}
is finitely exchangeable. \label{assp:ExchangeabilityRelevant}
\end{assumption}

\begin{assumption}[Constant Unpredictability of $Y$] For every realization of the standard covariates $(x_1, \dots, x_s)$ and every pair $n'>n$, the  sequence $(\widetilde{Y}^n_1, \dots, \widetilde{Y}_{2^{s+r_n}}^n)$ and the truncated sequence $(\widetilde{Y}^{n'}_1, \dots, \widetilde{Y}^{n'}_{2^{s+r_n}})$ are identical in distribution.  \label{assp:ConstantVarYRelevant} 
\end{assumption}

Our main model is otherwise unchanged---in particular, we allow the agent to disclose any of the $n-s$ nonstandard covariates, including those which are irrelevant.  
We show that our previous results extend so long as $\alpha_h < \alpha_r$.

\begin{proposition} \label{prop:Effective} Suppose Assumptions \ref{assp:ExchangeabilityRelevant} and \ref{assp:ConstantVarYRelevant} hold, where
$\alpha_h < \alpha_r$. Then for every covariate vector $\bold{x} \in \{0,1\}^\infty$ the expected value of context vanishes to zero as $n$ grows large.
\end{proposition}

The case where $\alpha_r < \alpha_h$ (violating the assumption of the result) corresponds to a setting in which the number of relevant covariates is so small that the agent can disclose all of them. For example, if a job candidate is convinced that only 10 nonstandard covariates are actually relevant for predicting his on-the-job ability, and all of these nonstandard covariates can be shared during a job interview, then our main results do not extend and we should think of the value of context as being potentially large. On the other hand, if the set of relevant covariates are low-dimensional relative to the total number of covariates, but are still too numerous to be fully revealed, then our main results do extend. 

This result suggests that whether human or black box evaluation is more appropriate should be determined in part based on whether the available signal is concentrated in a small number of covariates (favoring the human evaluator) or spread out across a large number of covariates (favoring the black box evaluator). The same application may transition between these regimes over time. For example, in a medical setting where black box diagnosis is highly accurate based on non-interpretable features of an image scan, it may not be possible to communicate sufficient information via any small number of covariates. But if the predictive features of the image are subsequently better understood and defined, then it may be that a small set of (newly defined) features does eventually capture all of the signal content, and can be fully disclosed in a conversation.

\subsection{Certain nonstandard covariates have known effects} \label{sec:KnownEffect}

Another possibility is that the agent knows how  certain nonstandard covariates are correlated with the type. 

\begin{example} The agent is a patient who resided around Chernobyl at the time of the Chernobyl nuclear disaster of 1986. The agent is being evaluated for potential thyroid conditions, and knows that this particular part of his history increases the probability of a thyroid condition.
\end{example}

Specifically suppose there is a set $K \subseteq \{1,\dots,n\}$ of covariate indices whose effects are known. The set $K$ includes the standard covariates, but possibly also includes some nonstandard covariates. Without loss, we can index these $x_1, \dots, x_K$. We weaken Assumption \ref{assp:Exchangeability} to the following:

\begin{assumption}[Exchangeability] For every realization of the covariates $(x_1, \dots, x_K)$,
\begin{equation}
(Y^n_1, \dots, Y^n_{2^n}) \equiv (f(x_1,\dots,x_K, x_{K+1}, \dots, x_n) : (x_{K+1},\dots,x_{n}) \in \{0,1\}^{n-\vert K \vert })
\end{equation}
is finitely exchangeable. \label{assp:ExchangeabilityKnown}
\end{assumption}

This assumption imposes exchangeability only over the nonstandard covariates whose effects are not ex-ante known. Clearly if $K$ is a strict superset of $\mathcal{S}$, then the expected value of context need not vanish under Assumptions \ref{assp:ConstantVarY} and \ref{assp:ExchangeabilityKnown}. A  simple example is the following.

\begin{example} Suppose there are no standard covariates, and $K= \{1\}$, i.e., the first nonstandard covariate has a known effect, where $f(\bold{x}_n) \sim U[-1,0]$ if $x_1=0$ and $f(\bold{x}_n) \sim U[0,1]$ if $x_1=1$. Suppose further that the agent's covariate vector satisfies $x_1=1$. Then the prior expectation of the agent's type is 0, but revealing $x_1 =1$ moves the posterior expectation to $1/2$. So the expected value of context does not vanish.
\end{example}

But if we modify the definition in (\ref{eq:CondExp}) to
\[\hat{y}^{f}_{\bold{x}_n}(A)= \mathbb{E}[Y \mid X_i = x_i \quad \forall i \in K\cup A]
\]
with $K$ replacing $\mathcal{S}$,
and again let $U^f_{\bold{x}_n}(A) = u\left(\hat{y}^{f}_{\bold{x}_n}(A),y\right)$, then the modified  expected value of context
\begin{equation*}
v(f,\bold{x}_n) = \max_{\substack{H \subseteq \mathcal{N}\backslash K \\ \vert H \vert \leq \alpha_h n}} U_{\bold{x}_n}^{f}(H) - U_{\bold{x}_n}^{f}(\varnothing)
\end{equation*}
evaluates the value of context beyond those covariates with known effects. The same proof shows that this expected value of context vanishes to zero as $n$ grows large. That is, beyond the value of context that is already clear to the agent based on private knowledge about his nonstandard covariates, the agent does not expect substantial additional gain from the remaining covariates.

\subsection{Information accumulates in $n$} \label{sec:NonConstant}

So far we have assumed that the predictability of $Y$ is constant in the number of covariates. This is not essential for our results. Suppose instead that 
\[y_n = f(x_1,\dots,x_n) +\varepsilon_n\]
where $\varepsilon_n$ is a mean-zero random variable that is independent of the covariates $(x_1,\dots,x_n)$. This describes a setting in which the covariates are not sufficient to reveal the agent's type, and there is  a residual unknown.

Our previous results extend directly when the distribution of $\varepsilon_n$ is the same for all $n$.
Another natural case is one in which $\Var(\varepsilon_n)$ decreases monotonically in $n$, with $\lim_{n\rightarrow \infty} \Var(\varepsilon_n)=0$; that is, in environments with a larger number of covariates, the agent's type is more predictable. In Appendix \ref{app:accum}
 we show that Theorem \ref{thm:Main} directly extends. We also show that the comparisons in parts (a) and (b) of Theorem \ref{prop:PreferHA} hold for sufficiently large $n$, under the following assumption:

 \begin{assumption} For each $n \in \mathbb{Z}_+$, let $\sigma_{\varepsilon,n}^2:=\Var(\varepsilon_n)$ and assume that $\frac{\varepsilon_n}{\sqrt{\sigma_{\varepsilon,n}^2}}$ admits a pdf, which we denote by $g_n$.     The sequence $\{g_n(0)\}_n$ is bounded.
\end{assumption}

Loosely speaking, our comparisons in Theorem \ref{prop:PreferHA} continue to hold so long as the variance of $\varepsilon_n$ does not increase too fast in $n$.

\subsection{Sufficient residual uncertainty} \label{sec:GeneralLemma}

 In this final section, we provide an abstract condition on the evaluator's learning environment, under which Theorem \ref{thm:Main} extends. 

For each $n$, let $\mathcal{D}_n$ denote the set of all disclosures respecting the human evaluator's capacity constraint, i.e., all pairs $(H, (x_i)_{i \in H})$ consisting of a set $H \subseteq \{s+1, \dots, n\}$ with $\lfloor \alpha_h \cdot n\rfloor$ or fewer nonstandard covariates, and values $(x_i)_{i \in H}$ for those covariates. Further define $\mathcal{D} = \cup_{n \geq 1} \mathcal{D}_n$ to be the set of all disclosures.   Similarly, for each $n$ let $\mathcal{F}_n$ be the set of all  type functions $f: \{0,1\}^n \rightarrow [-\overline{y},\overline{y}]$, and define $\mathcal{F} = \cup_{n \geq 1} \mathcal{F}_n$. An \emph{evaluation rule} is any family $\rho = (\rho_f)_{f \in \mathcal{F}}$ where each $\rho_f: \mathcal{D} \rightarrow[-\overline{y},\overline{y}]$ maps disclosures into evaluations for the given function $f$. Finally, fixing any update rule $\rho$, number of covariates $n$, and disclosure $d \in \mathcal{D}_n$, let 
\[Z_d^n = \rho_f(d)\]
be the random evaluation when $f$ is drawn from  $\mathcal{F}_n$ according to the agent's prior. 

We impose two assumptions below on the evaluation rule. The first says that the expected evaluation $Z_d^n$ is equal to the prior expected type $\mu \equiv \mathbb{E}[Y]$; the second says that the distribution of the evaluation concentrates on $\mu$ sufficiently fast as the number of hidden covariates $n$ grows large. Intuitively, the assumption requires that as the number of residual unknowns---i.e., the covariates which are predictive of the type, but are not revealed to the evaluator---grows large, the informativeness of any fixed disclosure becomes small.\footnote{In the limit with an uninformative disclosure, the distribution of the evaluation is degenerate at the prior expectation $\mu$ for any Bayesian updating rule.}

\begin{assumption}[Unbiased] $\E[Z_d^n]=\mu$ for every disclosure $d$. \label{assp:Unbiased}
    \end{assumption}

    \begin{assumption}[Fast Concentration]  For any sequence of feasible disclosures $(d_n)_{n \geq 1}$, 
     \[ Var(Z_{d_n}^n)=o\left(\frac{1}{K_n}\right)\]
     where $K_n=\sum_{j=0}^{\lfloor \alpha_h  n \rfloor}\binom{n-s}{j}$ is the number of unique sets $H \subseteq \{s+1, \dots, n\}$ with $\alpha_h n$ or fewer elements.
     \label{assp:Concentrate}
     \end{assumption}

These assumptions do not in general represent a weakening of our main model. Previously we studied the evaluation rule $\rho$ mapping each disclosure into the conditional expectation of the agent's type, and imposed Assumption \ref{assp:Exchangeability} on the agent's prior about $f$. In this model, the evaluation $Z_d^n$ for any disclosure $d =(H, (x_i)_{i \in H})$ could be represented  as a sample average consisting of $2^{n-s-\vert H \vert}$ elements. Assumption \ref{assp:Unbiased} is clearly satisfied (because the update rule is Bayesian), but one can select a sequence of disclosures $(d_n)$ such that $Var(Z_{d_n}^n) = \frac{1}{2^{n(1-\alpha_h)-s}}$ (see the proof of Theorem \ref{thm:Main} for details). Thus the speed of convergence demanded in Assumption \ref{assp:Concentrate} is not met when $\alpha_h$ is sufficiently large.

Nevertheless, Assumption \ref{assp:Concentrate} identifies the qualitative property of our main setting that gave us Theorem \ref{thm:Main}: residual uncertainty must have the power to overwhelm any information revealed through disclosure. Under these assumptions, our main result extends.

\begin{proposition}
Suppose Assumptions \ref{assp:Unbiased} and \ref{assp:Concentrate} hold. Then for every covariate vector $\bold{x} \in \{0,1\}^\infty$, the expected value of context vanishes to zero as $n$ grows large, i.e., 
\[\lim_{n \rightarrow \infty} V(n, \bold{x}_n)=0.\] \label{prop:non-bayesian} 
\end{proposition}
\vspace{-6mm}
This result also clarifies that neither the precise symmetry imposed by Assumption \ref{assp:Exchangeability}, nor the assumption of Bayesian updating in our main model, are crucial for our main result.

\section{Conclusion} \label{subsec:Discuss}

One argument against replacing human experts with algorithmic predictions is that no matter how many covariates are taken as input by the algorithm, the number of potentially relevant circumstances and characteristics is still more numerous. In cases where some important fact is missed by a human evaluator, it is often possible to correct this oversight. There is no such safety net with a black box algorithm.

This is a compelling narrative, yet our results suggest that it may be less important than it initially seems. When there is a large number of nonstandard covariates that may matter for the prediction problem, but the agent does not know how these nonstandard covariates impact the type, then the expected value of disclosing additional information is small---even when we assume that the agent can identify the most useful covariates to disclose, and that the claims about these covariates are taken at face value.

In contrast, if the agent has substantial prior knowledge about the predictive roles of the nonstandard covariates, then our conclusion will not be appropriate. In particular, if there is a ``low-dimensional'' set of covariates that predict the type and can be fully disclosed (as in Example \ref{ex:OnlyOne}), or if there is a known structural relationship between covariates and the type (as in Example \ref{ex:Structural}), then the expected value to disclosing additional information may be large. We thus view  our results as revealing a link betwen the value of targeting information acquisition (beyond simply conditioning on large quantities of information) and the extent of prior ``structural information'' about the numerous covariates that can be brought up as explanations.

We conclude with two alternative interpretations of our model and results.

\textbf{Online versus offline learning.} In our model, a key distinction between human and black box evaluation is that the human can adapt which covariates are acquired based on other properties of the agent, while the black box cannot. This is an appropriate comparison of human and black box evaluators as they currently stand: The black box algorithms used to make predictions about humans are usually  supervised machine learning algorithms which are pre-trained on a large data set. But new black box algorithms, such as large language models, blur this  distinction, and future evaluations (e.g., medical diagnoses) may be conducted by black box systems with which the agent can communicate.

From this more forward-looking perspective, our results can be understood as comparing the merits of online versus offline learning. That is, how valuable is it to have the evaluator dynamically acquire information given feedback from the agent? Our result suggests that this is not important in expectation. For example, Part (a) of Theorem \ref{prop:PreferHA} implies that an agent who cares about accuracy should prefer a supervised machine learning algorithm trained on a large number of covariates over a conversation with ChatGPT that reveals a smaller number of covariates. 

\textbf{Value of human supervision of algorithms.}
While we have interpreted the $s$ standard covariates as a small set of covariates acquired by the human evaluator, an alternative interpretation is that they are the initial inputs to an algorithm. In this case, the expected value of context quantifies the sensitivity of the algorithm's predictions to the addition of further relevant inputs, e.g.,  as identified by a human manager. This interpretation is particularly relevant when we consider accuracy as the objective, in which case the value of context tells us how wrong the algorithm is compared to if the algorithm could be retrained on additional relevant inputs. Theorem \ref{thm:Main} says that while in certain cases additional inputs would lead to a  substantially more accurate prediction, under our symmetry assumption on the agent's prior this will not typically be the case.

\appendix

\section{Proof of Generalization of Theorem \ref{thm:Main}}

In a change of notation relative to the main text, we subsequently use $\bold{X}_n$ to denote the agent's covariate vector and $Y$ to denote the agent's type (leaving $\bold{x}_n$ and $y$ to denote realizations of these random variables). Moreover, rather than supposing that $Y$ is deterministically related to $\bold{X}_n$ via a function $f$, let $(\bold{X}_n,Y)\sim P^n$ where $P^n$ is unknown. We replace Assumptions \ref{assp:Exchangeability} and \ref{assp:ConstantVarY} with the following weaker assumption.

\begin{assumption} Fix any realization of the standard covariates $\bold{x}_{\mathcal{S}} \in \{0,1\}^s$. There is an infinitely exchangeable sequence $(\widetilde{Y}_1, \widetilde{Y}_2, \dots)$ such that for every $n\in \mathbb{N}$, the sequence 
\[\left(\mathbb{E}[Y \mid (X_1,\dots,X_n) = (\bold{x}_{\mathcal{S}},\bold{x}_{-\mathcal{S}})]\right)_{\bold{x}_{-\mathcal{S}} \in \{0,1\}^{n-s}}\]
has the same distribution as $(\widetilde{Y}_1, \dots, \widetilde{Y}_{2^n})$. \label{assp:ExchangeabilityGeneral}
\end{assumption}

That is, permuting the labels and/or values of the nonstandard covariates does not change the joint distribution of the conditional expectations of $y$. When $y$ is degenerate conditional on $\bold{x}_n$, Assumption \ref{assp:ExchangeabilityGeneral} reduces to our previous two assumptions. We will prove the following generalization of Theorem \ref{thm:Main}.

\begin{theorem} \label{thm:MainGeneral}
Suppose Assumption \ref{assp:ExchangeabilityGeneral} holds. Then for every covariate vector $\bold{x} \in \{0,1\}^\infty$, the expected value of context vanishes to zero as $n$ grows large, i.e.,
	$\lim_{n\rightarrow\infty}V(n,\bold{x}_n)=0$.
 \label{thm:value}
\end{theorem}

Towards this, we will first prove the conclusion under a strengthening of Assumption \ref{assp:ExchangeabilityGeneral}, where exchangeability is replaced by an assumption that conditional expectations are i.i.d.\ across the different possible completions of the agent's covariate vector.

\begin{assumption} Fix any realization of the standard covariates $\bold{x}_{\mathcal{S}} \in \{0,1\}^s$. Then there is a distribution $F$ such that for every $n\in \mathbb{N}$, the conditional expectations 
\[\mathbb{E}[Y \mid (X_1,\dots,X_n) = (\bold{x}_{\mathcal{S}},\bold{x}_{-\mathcal{S}})]\sim_{iid} F \]
across all vectors $\bold{x}_{-\mathcal{S}} \in \{0,1\}^n$.  \label{assp:IID}
\end{assumption}

\begin{theorem} \label{thm:MainIID}
Suppose Assumption \ref{assp:IID} holds. Then for every covariate vector $\bold{x} \in \{0,1\}^\infty$, the expected value of context vanishes to zero as $n$ grows large, i.e.,
	$\lim_{n\rightarrow\infty}V(n,\bold{x}_n)=0$. 
\end{theorem}

Sections \ref{app:OutlineProof}-\ref{app:ConcludeProofMain} prove Theorem \ref{thm:MainIID}, and Section \ref{app:IIDtoExchange} shows that Theorem \ref{thm:MainIID} implies Theorem \ref{thm:MainGeneral}.

\subsection{Outline for Proof of Theorem \ref{thm:MainIID}} \label{app:OutlineProof}

Fix any realization $(x_1, \dots, x_s)$ of the agent's standard covariates. After observing $(x_1, \dots, x_s)$, the evaluator assigns positive probability to the $2^{n-s}$ covariate vectors whose first $s$ entries are equal to $(x_1, \dots, x_s)$. Let these covariate vectors be indexed by $\bold{x}^j$ where $j=1, \dots, 2^{n-s}$, and define
\[Y_j \equiv \mathbb{E}_{P^n}\left[Y \mid (X_1, \dots, X_n) = \bold{x}^j\right]\]
to be the (random) expected type given covariate vector $\bold{x}^j$.
By assumption that the marginal distribution over covariate vectors is uniform, the evaluator's posterior expectation of the agent's type after observing the agent's standard covariates is
\[\widehat{Y}(\varnothing,\bold{x}_n)= \frac{1}{2^{n-s}} \sum_{i=1}^{2^{n-s}} Y_j \equiv Z^n_\varnothing.\]
There are $K_n=\sum_{k=0}^{h_n}\binom{n-s}{k}$ subsets of  $\{s+1,\dots, n\}$ that contain $h_n$ or fewer elements. Enumerate these sets as $H_1, \dots, H_{K_n}$. For each $H_k$, let 
\[S_k = \left\{j \, : \, \bold{x}^j \in C_{H_k}(\bold{x}_n)\right\}\] be the set of indices for those covariate vectors  $\bold{x}^j$ that agree with the agent's covariate vector $\bold{x}_n$ in entries $(1,\dots,s) \cup H_k$ (where $C_{H_k}(\bold{x}_n)$ is as defined in (\ref{notation:C})). After observing the agent's nonstandard covariates in the set $H_k$, the evaluator's posterior expectation about the agent's type  is
 \[\widehat{Y}(H_k,\bold{x}_n) = \frac{\sum_{j\in S_k}Y_j}{|S_k|} \equiv Z_k.\]
Although the distributions of the random variables $Z_k$ vary across $n$, we suppress this dependence in what follows to save on notation.
The remainder of the proof proceeds by first showing that in expectation the possible increase in the evaluator's posterior expectation over the prior expectation $\mu \equiv \mathbb{E}[Y]$ is vanishing. 

\begin{proposition} \label{prop:Z}
$\lim_{n \rightarrow \infty} \E[\max_{1\le k\le K_n} Z_k - \mu]=0.$
\end{proposition}

\noindent This is subsequently strengthened to the statement that the expected maximum absolute difference between $ Z_k$ and $\mu$ converges to zero.

\begin{proposition}
$\lim_{n \rightarrow \infty} \E[\max_{1\le k\le K_n} \vert Z_k -\mu \vert]=0.$
\label{lem_abs}
\end{proposition}

\noindent And finally we apply the above proposition to demonstrate the  conclusion of the theorem, i.e., that 
\[\lim_{n \rightarrow \infty} V(n) = \lim_{n \rightarrow \infty} \mathbb{E}\left[\max_{1 \leq k \leq K_n} u(Z_k,Y)\right] - \mathbb{E}\left[u\left( Z^n_\varnothing,Y\right) \right] = 0\]
Thus in expectation the possible increase in the agent's payoff also vanishes. We suppress dependence of $V$ on the covariate vector $\mathbf{x}_n$ in what follows, writing simply $V(n)$.

\subsection{Proof of Proposition \ref{prop:Z}}

\emph{Statement of the proposition: $\lim_{n \rightarrow \infty} \E[\max_{1\le k\le K_n} Z_k-\mu]=0.$}
\bigskip

\noindent The quantity $\E[\max_{1\le k\le K_n} Z_k]$ is the expected first-order statistic of a sequence of non-i.i.d.\ variables $Z_1, \dots, Z_{K_n}$. The proof is organized as follows. In Sections \ref{proof:Step1} and \ref{proof:Step2}, we define  i.i.d.\ variables $Z^{iid}_k$ with the property that
\begin{equation} 
\mathbb{E}\left[\max\{Z_1, \dots, Z_{K_n}\}\right] \leq  \mathbb{E}\left[\max\{Z^{iid}_1, \dots, Z^{iid}_{K_n}\}\right].
\end{equation}
In Sections \ref{proof:Step3} and \ref{proof:Step4}, we show that the RHS of the above display converges to $\mu$ as $n$ grows large.

\subsubsection{Replacing $Z_k$'s with independent variables $Z_k^{ind}$} \label{proof:Step1}

In general, disclosures $k$ and $k'$ may lead to  posterior expectations $Z_k$ and $Z_{k'}$ that are correlated due to the presence of the same $Y_i$'s across the different sample averages. We first show that replacing these $Z_k$'s with properly defined independent random variables weakly increases the value of context.

\begin{definition}
    For each $1\leq k\leq K_n$  define
 \begin{equation} \label{def:Zind}
 Z_k^{ind}=\frac{\sum_{j=1}^{\vert S_k \vert} Y_j^k}{|S_k|}
 \end{equation}
where $Y_j^k\sim_{iid} F$, so that each $Z_k^{ind}$ has the same distribution as $Z_k$, but the vector $(Z_1^{ind}, \dots, Z_K^{ind})$ is mutually independent. 
\end{definition}
\begin{lemma} Let 
\[V_n \equiv \E[\maxx{Z_1,...,Z_{K_n}}]\]
and
\[V^{ind}_n \equiv \E[\maxx{Z_1^{ind},...,Z_{K_n}^{ind}}].\]
Then $V_n \leq V^{ind}_n$ for all $n \in \mathbb{Z}_+$. \label{prop:Ind}
\end{lemma}

\begin{proof}
Throughout  we use  $X\succeq Y$  to mean that the distribution of $X$ first-order stochastically dominates the distribution of $Y$.  

\begin{sublemma}
	Let $X_1$,..., $X_{Q}$,$W$ be  a sequence of real-valued random variables (not necessarily i.i.d.). Let $a_1>a_2>...>a_{Q-1}>a_{Q}>0$ be  a sequence of positive constants. Further, let $Y_1,...,Y_{Q}$ be i.i.d.\ random variables, independent of $(X_1, \dots, X_Q,W)$. Define
\begin{align*}
 M_C&=\max_{i\in\{1,...,Q\}}\{X_i+a_iY_1\}\\
 M_I&=\max_{i\in\{1,...,Q\}}\{X_i+a_iY_i\}
 \end{align*}
 Then $M_I \succeq M_C$ and $\max\{M_I,W\} \succeq \max\{M_C,W\}$.
 \label{fosd_g}
\end{sublemma}

\begin{proof}
	
For $q\in\{1,...,Q\}$ define:	
 \begin{align*}
     M_C^{q} & =\max\left\{ \max_{i\in\{1,...,q-1\}}\{X_i+a_iY_1\}, X_q + a_q Y_1 \right\} \\
 \widetilde{M}_C^{q} & =\max \left\{\max_{i\in\{1,...,q-1\}}\{X_i+a_iY_1\},X_q+a_qY_q\right\}
 \end{align*}
 so that $M_C^q$ is the maximum of the first $q$ terms in $M_C$, and $\widetilde{M}_C^q$ replaces $Y_1$ in the $q$-th term of $M_C^q$ with $Y_q$. We first demonstrate an analogue of the desired conclusions for $M_C^q$ and $\widetilde{M}_C^q$.

\begin{sublemma}
$\widetilde{M}_C^q \succeq M_C^q$ and 
$\max\{\widetilde{M}_C^q,W\} \succeq \max\{M_C^q,W\}$.
\label{fosd_k}
\end{sublemma}

\begin{proof}

Without loss of generality set $a_q=1$. We'll first  show that $\widetilde{M}_C^q\succeq M_C^{q}$. To establish first-order stochastic dominance, we need to show that for all $t\in \R$ it holds that 
 \[\P(M_C^{q}\le t)-\P(\widetilde{M}_C^q\le t)\ge 0\]
 For each $i\in\{1,...,q-1\}$ define the event  \[B_i:=\{X_q+Y_1> X_i+a_iY_1\}\equiv \left\{Y_1<\frac{1}{a_i-1}(X_q-X_i)\right\}.\]
 Further let 
 \[B=\bigcap_{i=1}^q B_i=\left\{Y_1<\min_{i\in\{1,...,q-1\}}\frac{1}{a_i-1}(X_q-X_i)\right\}\]
 be the event that $X_q + Y_1$ achieves the maximum among $\{X_i+a_iY_1\}_{i=1}^q$. We'll show that the FOSD rankings in Sublemma \ref{fosd_k} hold both on event $B$ and also on its complement $B^c$.
 
 Define
 \[\widetilde{B}:=\left(Y_{q}<\min_{i\in\{1,...,q-1\}}\left\{\frac{1}{a_i-1}(X_q-X_i)\right\}\right)\]
 to be the event that $X_q + Y_q$ achieves the maximum among $\{X_i + a_i Y_q\}_{i=1}^q$.
 Then
 \begin{align*} \widetilde{M}_C^q|B& \succeq (X_q+Y_{q})|B \\
 &\stackrel{d}{=}X_q|B+Y_{q} && \text{since }Y_q \indep (X_1, \dots, X_q, Y_1)\\
 & \succeq X_q|B+Y_{q}|\Tilde{B} && \text{since } Y_q \succeq Y_q \mid \widetilde{B}\\
 & \stackrel{d}{=} X_q|B+Y_1|B && \text{since } Y_1 \mid B \stackrel{d}{=} Y_q \mid \widetilde{B}\\
 & \stackrel{d}{=} (X_q+Y_1)|B \stackrel{d}{=} M
 _C^q|B
 \end{align*}
 Thus $\widetilde{M}_C^q|B\succeq M_C^q|B$.

Now consider the event $B^c$,  on which $X_q+Y_1$ does not achieve the maximum among $\{X_i+a_iY_1\}_{i=1}^q$. Then either $X_1 + Y_q \leq \max\{X_i + a_i Y_1\}_{i=1}^{q-1}$, in which case $\widetilde{M}_C^q = M_C^q$, or $X_1 + Y_q > \max\{X_i + a_i Y_1\}_{i=1}^{q-1}$, in which case $\widetilde{M}_C^q > M_C^q$. So
 \[\widetilde{M}_C^q|B^c\succeq\maxx{X_1+a_1Y_1,...,X_{q-1}+a_{q-1}Y_1}|B^c\stackrel{d}{=} M_C^q|B^c.\]
and hence $\widetilde{M}_C^q|B^c\succeq M_C^q|B^c$.

 Now we show that $\max\{\widetilde{M}_C^q,W\} \succeq \max\{M_C^q,W\}$. For any realization $w$ of $W$, let $X^w_i$ denote the conditional random variable $X_i|W=w$. Define $M_C^{q,w}$ and $\widetilde{M}_C^{q,w}$ identically to $M_C^q$ and $\widetilde{M}_C^q$, replacing each $X_i$ by $X_i^w$. Then by 
 independence of $W$ and $(Y_1, \dots, Y_q)$, the distribution of $\max\{M_C^{q,w},w\}$ is identical to that of $\max\{M_C^q,W\}|W=w$, and the distribution of $\max\{\widetilde{M}_C^{q,w},w\}$ is identical to that of $\max\{\widetilde{M}_C^q,W\}|(W=w)$.
 
Applying the first part of this sublemma to $M^{q,w}_C$ and $\widetilde{M}_C^{q,w}$, we conclude that $M^{q,w}_I \succeq M_C^{q,w}$. Since $\max\{.,w\}$ is an increasing convex function, it preserves the first-order stochastic dominance relation and hence $\max\{\widetilde{M}_C^q,W\}|(W=w) \succeq \max\{M^q_C,W\}|(W=w)$. This argument holds pointwise for all $w$ so $\max\{\widetilde{M}_C^q,W\} \succeq \max\{M_C^q,W\}$ as desired. 
\end{proof}

We now complete the proof that $\max\{M_C,W\}\succeq \max\{M_I,W\}$. From similar (omitted) arguments it follows that $M_I \succeq M_C$. For each $ q\in \{ 1, \dots, Q-1\}$ define
\[\widehat{M}_C^q = \max \left\{ \max\{ X_i + a_i Y_1 \}_{i=1}^q, \max\{ X_i + a_i Y_i \}_{i=q+1}^Q, W\right\}\]
observing that $\max\{M_I,W\} = \widehat{M}_C^1$ and that $\widehat{M}_C^Q \succeq \max\{M_C,W\}$ (by Sublemma \ref{fosd_k}). Moreover, for each $q \in \{1, \dots, Q-1\}$,
\begin{align*}
    \widehat{M}_C^q  & = \max \left\{ M_C^q, W^q \right\} \\
\widehat{M}_C^{q-1} & =  \max\{ \widetilde{M}_C^q,  W^q \}
\end{align*} where 
$W^q = \max\left\{\max\{X_i + a_i Y_i\}_{i=q}^Q,W\right\}$
is independent of $(Y_1, \dots, Y_{q-1})$. 
So applying Sublemma \ref{fosd_k}, $\widehat{M}_C^{q-1} \succeq \widehat{M}_C^q$ as desired.  \end{proof}

Finally, we use Sublemma \ref{fosd_g} to establish Lemma \ref{prop:Ind}, i.e.,  the expected value of context weakly increases if we make the $Y$'s within different disclosures independent. We will prove this iteratively.  For arbitrary $n \in \mathbb{N}$, define the random variable
\[M = \max \{ Z_1, \dots, Z_{K_n}\} = \max \left\{ \frac{\sum_{j \in S_1} Y_j}{\vert S_1 \vert}, \dots, \frac{\sum_{j \in S_{K_n}} Y_j}{\vert S_{K_n} \vert} \right\}.\]
Fix any $Y_i$. We will show that replacing  $Y_i$ across different sample averages with independent copies of this random variable leads to a FOSD increase in the distribution of $M$. 

Let $I = \{ k : i \in S_k\}$ be the set of indices of sample averages which contain $Y_i$. Then we can rewrite the previous display as
\begin{align*}
\max\left\{ \, \max_{k \in I} \frac{\sum_{j \in S_k} Y_k}{\vert S_k \vert}, \, \max_{k \notin I}
\frac{\sum_{j \in S_k} Y_k}{\vert S_k \vert}\right\}
\end{align*}
or
\begin{equation} \label{eq:Correlated}\max\left\{ \max_{k \in I} \left\{ X_k + \frac{1}{\vert S_k \vert} Y_i \right\}, W\right\}
\end{equation}
where  $X_k \equiv \frac{1}{\vert S_k \vert} \sum_{j \in S_k, j \neq i} Y_j$  for each $k \in I$, and $W \equiv \max_{k \notin I}
\frac{\sum_{j \in S_k} Y_k}{\vert S_k \vert}$. Because $(Y_1, \dots, Y_{K_n})$ are mutually independent, $Y_i$ is independent of each $X_k$ and $W$. So applying Lemma \ref{fosd_g}, the random variable in (\ref{eq:Correlated}) has a distribution that is first-order stochastically dominated by the distribution of 
\[\max\left\{ \max_{k \in I} \left\{ X_k + \frac{1}{\vert S_k \vert} Y_i^k \right\}, W\right\}\]
as desired. Since $Y_i$ is arbitrary, this concludes the proof.

\end{proof}

\subsubsection{Replacing $Z_k^{ind}$ with i.i.d.\ Variables $Z_k^{iid}$} \label{proof:Step2}

The variables $Z_1^{ind}, \dots, Z_{K_n}^{ind}$ are sample averages of unequal sizes ranging between $2^{n-s-h_n}$ and $2^{n-s}$ elements. We  next show that replacing each of these variables with a sample average of $2^{n-s-h_n}$ elements (the smallest size) weakly increases the value of context.  

\begin{definition} For each $1\leq k \leq K_n$ define
 \begin{equation} \label{def:Ziid} Z_k^{iid}=\frac{\sum_{j=1}^{2^{n-s- h_n}} Y_j^k}{2^{n-s-h_n}}
 \end{equation}
to be the analogue of $Z_k^{ind}$ with $2^{n-s-h_n}$ elements instead of $\vert S_k \vert \geq 2^{n-s-h_n}$, so that the variables $Z_1^{iid}, \dots, Z_{K_n}^{iid}$ are iid. 
\end{definition}
\begin{lemma} Let 
\[V^{iid}_n \equiv \E\left[\maxx{Z^{iid}_1, \dots, Z^{iid}_{K_n}}\right].\]
Then $V^{ind}_n \le V^{iid}_n$ for all $n \in \mathbb{Z}_+$.
\end{lemma}

\begin{proof}
We use the following result.

\begin{sublemma} Suppose $Y_1, Y_2, \dots, Y_n$ are independent and identically distributed random variables, and define $\overline{Y}_n = \frac1n \sum_{i=1}^n Y_i$ to be their sample average. Let $n' < n$ and define $\overline{Y}_{n'} = \frac{1}{n'} \sum_{i=1}^{n'} Y_i$. Then the distribution of $\overline{Y}_{n'}$ is a mean preserving spread of the distribution of $\overline{Y}_n$.
\end{sublemma}

\begin{proof} First observe that $\mathbb{E}[Y_j \mid \overline{Y}_n] = \overline{Y}_n$ for any $j = 1,\dots,n$, since 
\[\overline{Y}_n = \mathbb{E}[\overline{Y}_n \mid \overline{Y}_n] = \frac1n \sum_{i=1}^n \mathbb{E}[Y_i \mid \overline{Y}_n] = \mathbb{E}[Y_j \mid \overline{Y}_n]\]
where the final equality follows by assumption that the $Y_i$'s are iid. Then
\begin{align*}
\mathbb{E}[\overline{Y}_{n'} \mid \overline{Y}_n] &= \frac{1}{n'} \sum_{i=1}^{n'} \mathbb{E}[Y_i \mid \overline{Y}_n] = \frac{1}{n'} \sum_{i=1}^{n'} \overline{Y}_n = \overline{Y}_n
\end{align*}
and the distribution of $\overline{Y}_{n'}$ is a mean-preserving spread of the distribution of $\overline{Y}_n$ as desired. 
\end{proof}

This lemma implies that each $Z_k^{iid}$ second-order stochastically dominates $Z_k^{ind}$ (since $\vert S_k \vert \geq 2^{n-s - h_n}$ for all $k$). The desired result then follows by Jensen's inequality, since the entries of $(Z_1^{ind}, \dots, Z_K^{ind})$ are (by construction) independent and the maximum is a convex function.
\end{proof}

\subsubsection{ Asymptotic Normality} 
\label{proof:Step3}

\begin{lemma} Let 
\[V^N_n\equiv \E\left[\maxx{Z_1^N, \dots, Z_{K_n}^N}\right]\]
where $Z^N_k\sim \mathcal{N}\left(\mu,\frac{1}{2^{n-s-h_n}}\right)$. Then 
	$\lim_{n \rightarrow \infty} \vert V^{iid}_n - V^N_n \vert = 0.$
\end{lemma}

\begin{proof}
Without loss of generality, let $Var(Y_j^k)=1$.\footnote{If $Var(Y_j^k)=0$, the statement of Theorem \ref{thm:Main} holds trivially.} First observe that
\[\sqrt{2^{n-s-h_n}} \cdot V^{iid}_n = \mathbb{E}\left[\maxx{\widetilde{Z}^{iid}_1, \dots, \widetilde{Z}^{iid}_{K_n}}\right]\]
where each
\[\widetilde{Z}^{iid}_k = \frac{1}{\sqrt{2^{n-s-h_n}}} \sum_{i=1}^{2^{n-s-h_n}} Y^k_i.\]
Similarly we can write
\[\sqrt{2^{n-s-h_n}} \cdot V^N_n = \mathbb{E}\left[\maxx{\widetilde{Z}^N_1, \dots, \widetilde{Z}^N_{K_n}}\right]\]
where each
\[\widetilde{Z}^N \sim_{iid} \mathcal{N}\left(\mu, 1\right).\]
 When the assumptions for 
Corollary 2.1 from \cite{chernozhukov2013gaussian} are met (to be verified momentarily), we can conclude that
\[\rho\left(\maxx{\widetilde{Z}^{iid}_1, \dots, \widetilde{Z}^{iid}_{K_n}},\maxx{\widetilde{Z}^N_1, \dots, \widetilde{Z}^N_{K_n}}\right) \rightarrow 0
\]
where $\rho$ denotes Kolmogorov distance. Thus also 
\begin{equation} \label{conv:KD} \rho(M^{iid}_n, M^N_n) \rightarrow 0
\end{equation}
where
\begin{align*}
    M^{iid}_n & = \frac{1}{\sqrt{2^{n-s-h_n}} } \maxx{\widetilde{Z}^{iid}_1, \dots, \widetilde{Z}^{iid}_{K_n}} \\
    M^N_n & = \frac{1}{\sqrt{2^{n-s-h_n}} } \maxx{\widetilde{Z}^N_1, \dots, \widetilde{Z}^N_{K_n}}
\end{align*}

By assumption, each $Y_i^k$ is supported on $\left[-\overline{y},\overline{y}\right]$ for some finite $\overline{y}$. This implies  $\vert M_n^{iid}\vert \leq \overline{y}$ for all $n$, so the sequence $(M_n^{iid})_n$ is uniformly integrable. The convergence in (\ref{conv:KD}) thus implies
\[\lim_{n \rightarrow \infty} \left\vert \mathbb{E}\left[M_n^{iid}\right] - \mathbb{E}\left[M_n^N\right] \right\vert = \lim_{n \rightarrow \infty} \vert V^{iid}_n - V^N_n \vert = 0 \]
as desired. 
 
It remains to verify that the conditions of Corollary 2.1 from \cite{chernozhukov2013gaussian} are met. This follows from the assumption that $Y_j^k$'s are uniformly bounded, and the observation that
\[\frac{\log(K_n\cdot2^{n-s-h_n})^7}{2^{(1-c)(n-s-h_n)}}\xrightarrow{n\rightarrow\infty}0\]
for any $c\in (0,1)$, since $K_n = \sum_{j=0}^{h_n} {n-s \choose j} \le 2^{n-s}$ by the Binomial Theorem and $\alpha_h<1$.
\end{proof}

\subsubsection{Upper Bound for Expected Maximum of Gaussians} \label{proof:Step4} Finally by \citet{Berman1964}, which provides an upper bound for the expected maximum of independent Gaussian random variables
\[V^N_n\le \frac{1}{2^{n-s-h_n}} \cdot 2\sqrt{\log(K_n)}\le\frac{1}{2^{n(1-\alpha_h)-s}} \cdot 2\sqrt{n}\]
where the final expression converges to zero as $n \rightarrow \infty$ by assumption that $\alpha_h < 1$. Since clearly also $\lim_{n \rightarrow \infty} \mathbb{E}[\max_{1 \leq k \leq K_n} Z_k - \mu] \geq 0$, this concludes the proof of Proposition \ref{prop:Z}.

\subsection{Proof of Proposition \ref{lem_abs}}

\emph{Statement of the proposition: $\lim_{n \rightarrow \infty} \E[\max_{1\le k\le K_n} \vert Z_k-\mu \vert ]=0.$}
\bigskip

In an abuse of notation, let $Z_k \equiv Z_k - \mu$ denote de-meaned sample average. By rewriting the max within the expectation we obtain
	\begin{align*}\E\left[\max_{1\le k\le K_n}|Z_k|\right] & =\E\left[\max\left\{\max_{1\le k\le K_n}Z_k,-\min_{1\le k\le K_n}Z_k\right\}\right] \\
 & \le \E\left[\max\left\{\max_{1\le k\le K_n}\left\{Z_k\right\},0\right\}\right]+\E\left[\max\left\{-\min_{1\le k\le K_n}\left\{Z_k \right\},0\right\}\right] \end{align*}
 We will show that each term of this final expression converges to zero. Observe that
 \begin{equation} \label{eq:FirstTerm}
 \E\left[\max\left\{\max_{1\le k\le K_n}\left\{Z_k\right\},0\right\}\right]  = \mathbb{P}\left(\max_{1\le k\le K_n}Z_k \geq 0\right) \cdot \mathbb{E}\left[ \max_{1\le k\le K_n}Z_k \mid \max_{1\le k\le K_n}Z_k \geq 0 \right] 
 \end{equation}
Moreover,
\begin{align*}
 \mathbb{E}\left[ \max_{1 \leq k \leq K_n } Z_k \right] & = \mathbb{P}\left(\max_{1\le k\le K_n}Z_k \geq 0\right) \cdot \mathbb{E}\left[ \max_{1\le k\le K_n}Z_k \mid \max_{1\le k\le K_n}Z_k \geq 0 \right]  \\
 & \quad \quad +  \mathbb{P}\left(\max_{1\le k\le K_n}Z_k < 0\right) \cdot \mathbb{E}\left[ \max_{1\le k\le K_n}Z_k \mid \max_{1\le k\le K_n}Z_k < 0 \right]
\end{align*}
so 
\begin{align} \label{eq:ThirdTerm} 
 \mathbb{P}\left(\max_{1\le k\le K_n}Z_k \geq 0\right) & \cdot \mathbb{E}\left[ \max_{1\le k\le K_n}Z_k \mid \max_{1\le k\le K_n}Z_k \geq 0 \right]   =  \nonumber \\ 
 & = \mathbb{E}\left[ \max_{1 \leq k \leq K_n } Z_k \right] - \mathbb{P}\left(\max_{1\le k\le K_n}Z_k < 0\right) \cdot \mathbb{E}\left[ \max_{1\le k\le K_n}Z_k \mid \max_{1\le k\le K_n}Z_k < 0 \right]
\end{align}
From Lemma \ref{prop:Z}, \begin{equation} \label{eq:SecondTerm}\lim_{n \rightarrow\infty} \mathbb{E}\left[ \max_{1 \leq k \leq K_n } Z_k \right] = 0.
\end{equation}
Moreover, we showed in Section \ref{proof:Step1} that  the distribution of $(Z^{ind}_1, \dots, Z^{ind}_{K_n})$ first-order-stochastically-dominates that of $(Z_1, \dots, Z_{K_n})$, so
\begin{align*}
 \mathbb{P}\left(\max_{1\le k\le K_n}Z_k < 0\right)  \leq   \mathbb{P}\left(\max_{1\le k\le K_n}Z^{ind}_k < 0\right) \leq \prod_{1 \leq k \leq K_n} \mathbb{P}(Z_k^{ind} < 0)
\end{align*}
which converges to zero as $n$ grows large since each $\mathbb{P}(Z_k^{ind}<0) < 1$. Finally, 
\begin{equation} \label{eq:FifthTerm}\mathbb{E}\left[ \max_{1\le k\le K_n}Z_k \mid \max_{1\le k\le K_n}Z_k < 0 \right]  \in \left[-\overline{Y},  \overline{Y}\right]
\end{equation}
uniformly across $n$. 
Putting together (\ref{eq:FirstTerm}) - (\ref{eq:FifthTerm}) we have that
\[ \lim_{n \rightarrow \infty} \E\left[\max\left\{\max_{1\le k\le K_n}\left\{Z_k\right\},0\right\}\right] = 0\]
as desired. The argument that \[\lim_{n \rightarrow \infty} \E\left[\max\left\{-\min_{1\le k\le K_n}\left\{Z_k \right\},0\right\} \right] = 0\] follows identically, observing that Proposition \ref{prop:Z} is satisfied for $\widetilde{Y} \equiv -Y$, and that 
\[-\min_{1 \leq k \leq K_n} Z_k = \max_{1 \leq k \leq K_n}-\frac{\sum_{j\in S_k}Y_j}{|S_k|} =  \max_{1 \leq k \leq K_n}\frac{\sum_{j\in S_k}\widetilde{Y}_j}{|S_k|}.\]

\subsection{Concluding the proof of Theorem \ref{thm:MainIID}} \label{app:ConcludeProofMain}

	  Recall that $Z^n_\varnothing \equiv \frac{1}{2^{n-s}}\sum_{j=1}^{2^{n-s}} Y_j$ denotes the (random) posterior expectation when the agent chooses not to disclose any nonstandard covariates. Clearly $V(n) \geq 0$ (since the agent can always choose to disclose nothing). Also	
	\begin{align}
		V(n) & = \mathbb{E}\left[ \max_{1 \leq k \leq K_n} u(Z_k,Y)\right] -\E
	\left[u(Z^n_\varnothing,Y)
		\right] \nonumber \\
		& \le\E\left[\max_{1\le k\le K_n}|u(Z_k,Y)-u(Z^n_\varnothing,Y)|\right] 
	\label{eq_b0}
	\end{align}
Each absolute difference $|u(Z_k,Y)-u(Z^n_\varnothing,Y)|$ can be bounded from above using the triangle inequality	
	\begin{equation}
    |u(Z_k,Y)-u(Z^n_\varnothing,Y)|\le |u(Z_k,Y)-u(\mu,Y)|+ |u(\mu,Y)-u(Z^n_\varnothing,Y)| \label{eq_b1}
 \end{equation}
	
	Since $u$ is by assumption Lipschitz continuous in the first argument, there is a constant $B$ such that 
\begin{equation}
		|u(z_k,y)-u(\mu,y)|\le B|z_k-\mu|
	\label{eq_b2}
	\end{equation}
	and
\begin{equation}
		|u(\mu,y)-u(z_\varnothing,y)|\le B|z_\varnothing-\mu|\label{eq_b3}
	\end{equation}
 for any realizations $z_k$ and $z_\varnothing$ of $Z_k$ and $Z^n_\varnothing$. Combining equations \ref{eq_b0}-\ref{eq_b3} 
	 we get 
	 \[V(n)\le B\left(\E\left[\max_{1\le k\le K_n}|Z_{k}-\mu|\right]+\E\left[|Z^n_\varnothing-\mu|\right]\right) \]
	 
	Clearly $\E[Z^n_\varnothing]=\mu$. Moreover, by assumption that each $Y$ is uniformly bounded above and below, the sequence $(Z^n_\varnothing)$ is uniformly integrable. It follows from the Law of Large Numbers that
	 \[\lim_{n \rightarrow \infty} \E[|Z^n_\varnothing-\mu|]=0\]
  Finally, $\lim_{n \rightarrow \infty} \E\left[\max_{1\le k\le K_n}|Z_{k}-\mu|\right]=0$ follows directly from Lemma \ref{lem_abs}. So the RHS of \ref{eq_b1} converges to zero, implying $V(n) \rightarrow 0$ as desired. 

	\subsection{Theorem \ref{thm:MainIID} implies Theorem \ref{thm:MainGeneral}} \label{app:IIDtoExchange}

In an abuse of notation, let $P^n \sim F$ mean that $Y_{\bold{x}_n} \sim_{iid} F$ across all covariate vectors $\bold{x}_n$. We have already shown in Theorem \ref{thm:Main} that  
$\lim_{n \rightarrow \infty} \mathbb{E}_{P^n \sim F} (v_n(P))=0$ for any distribution $F$.  Now suppose instead that Assumption \ref{assp:Exchangeability} is satisfied. By de Finetti's theorem, there exists a set $\Theta$, family of conditional measures $(\pi_\theta)_{\theta \in \Theta}$, and measure $\nu \in \Delta(\Theta)$ such that
\begin{align*}
    V(n,\bold{x}) = \int_\Theta \mathbb{E}_{P^n \sim F_\theta}(v_n(P,\bold{x}_n)) d\nu(\theta)
\end{align*}
where the inner expectation converges to zero for every $\theta$ by Theorem \ref{thm:MainIID}. By assumption that $u$ is Lipschitz continuous on a compact domain, there exist $\underline{u}$ and $\overline{u}$ such that $u(\hat{y},y) \in [\underline{u},\overline{u}]$ for all $(\hat{y},y)$. So $\mathbb{E}_{P^n \sim F_\theta}(v_n(P,\bold{x}_n))$ is pointwise bounded above by $ \overline{u}-\underline{u}$, and the Dominated Convergence Theorem implies $\lim_{n \rightarrow \infty} V(n,\bold{x})=0$, as desired.

\subsection{Proof of Theorem \ref{prop:PreferHA}}
\label{app:GenAccuracy}
Throughout the proof we set $s=0$, $\mu=0$ and $\sigma^2=\E(Y_i^2)=1$ without loss of generality. We'll start by demonstrating that the stated results hold asymptotically (i.e., for large enough $n$) and subsequently prove that the bound in (\ref{eq:N}) is sufficient.

\emph{(a)} As before let $B_n \subseteq \{1, \dots, 2^{n}\}$ index those $2^{n-b_n}$ covariate vectors that agree with the agent's covariate vector for all covariates in $B$. Then the black box evaluator's posterior expectation is the sample average
\[Z_B^n = \frac{1}{2^{n-b_n}} \sum_{j \in B_n} Y_j.\] We will show  that 	\begin{align*}
\Delta(n)& \equiv\E[\phi(Z_B^n)]-\E\left[\max_{1\le k\le K_n} \phi(Z_k)\right]\\
		&=\E[\phi(Z_B^n)-\phi(0)] - \E\left[\max_{1\le k\le K_n} \phi(Z_k)-\phi(0)\right]>0
	\end{align*}
for large enough $n$.
	
We start by analyzing the first difference $\E[\phi(Z_B^n)-\phi(0)]$. Using Taylor's expansion we get
\[\E[\phi(Z_B^n)-\phi(0)]= \E[\phi'(0)Z_B^n]+\E\left[\frac{\phi''(\Tilde{Z})}{2}(Z_B^n)^2\right]\]
for some $\Tilde{Z}\in [0,Z_B^n]$. Note that $\E[Z_B^n]=\E[Y]=0$. Moreover, $\phi''(\Tilde{Z})\ge c_1>0$ for some $c_1$, since $\phi$ is strictly convex. Thus
\begin{equation} \label{eq:LB}
\E[\phi(Z_B^n)-\phi(0)]\ge  c_1\E[(Z_B^n)^2]=\frac{c_1}{2^{(1-\alpha_b)n}}
\end{equation}

Next turn to $\E[\max_{1\le k\le K_n} \phi(Z_k)-\phi(0)]$.
	For each term inside the maximum we have that 
	\begin{equation} \label{eq:UB}
 \phi(Z_k)-\phi(0)\le c_2|Z_k|
 \end{equation}
	where the latter inequality follows from the fact that $\phi'$ is continuous on a compact set, and hence bounded by some $c_2\ge 0$. Thus 
	\[\E\left[\max_{1\le k\le K_n}\phi(Z_k)-\phi(0)\right]\le c_2\E\left[\max\{|Z_k|\}\right]\]
	
	From our proof of Proposition \ref{lem_abs} it follows that 
	\[\E[\max\{|Z_k|\}]\le \frac{1}{2^{(1-\alpha_h)n-1}}\sqrt{\log(K_n)}\]
	And, thus	\[\E\left[\max_{1\le k\le K_n}\phi(Z_k)-\phi(0)\right]\le c_2\frac{1}{2^{(1-\alpha_h)n-1}}\sqrt{\log(K_n)}\]

	Combining the bounds from steps 1 and 2 we get
		\[\Delta(n) \ge c_1\frac{1}{2^{(1-\alpha_b)n}}-2c_2\frac{1}{2^{(1-\alpha_h)n}}\sqrt{\log(K_n)} \]
	The RHS is positive for all large $n$ if and only if 
	\[\frac{2^{(1-\alpha_h)n}}{2^{(1-\alpha_b)n}}\xrightarrow{n\rightarrow\infty}\infty\]
since $\sqrt{\log(K_n)}$ has sub-exponential but non-constant asymptotics. This condition is satisfied if and only if $\alpha_b>\alpha_h$.

\emph{(b)} Since $-\phi$ is convex, the above arguments apply to show that 
\[\E[\phi(Z_B^n)-\phi(0)]\le -\frac{c_1}{2^{(1-\alpha_b)n}}\]
for some $c_1>0$, while
\begin{align*}
\E\left[\min_{1\le k\le K_n}\phi(Z_k)-\phi(0)\right] &= -\E\left[\max_{1\le k\le K_n}-\phi(Z_k)-(-\phi(0))\right] \\
&\ge -c_2\frac{1}{2^{(1-\alpha_h)n-1}}\sqrt{\log(K_n)}
\end{align*}
for some $c_2>0$. The desired conclusion follows.

\bigskip
Finally, we show that the bound in (\ref{eq:N}) is sufficient for the comparison in part $(a)$ of the result (with identical computations applying to part $(b)$). Suppose $\phi(.)$ is strictly convex and denote $C=\frac{2c_2}{c_1}$, where $c_1,c_2>0$ are the constants used above, respectively reflecting $\phi$'s lowest degree of convexity ($c_1=\inf_{y\in[-\overline{y},\overline{y}]}|\phi''(y)|$) and largest growth rate ($c_2=\sup_{y\in[-\overline{y},\overline{y}]}|\phi'(y)|$). Following the proof of part (a), the Black Box is preferred if
\[2^{(\alpha_b-\alpha_h)n}>C\sqrt{\log(K_n)}\]
Since $K_n\le 2^n$, this inequality is satisfied if
\[(\alpha_b-\alpha_r)n-\frac{1}{2}\log_2(n)>\log_2(C)\]
The above inequality implicitly defines a  threshold on the sufficient number of covariates $N(C)$, exceeding which the Black Box is preferred. 
\subsection{Result Extending Theorem \ref{prop:PreferHA} Part (a)}
\label{app:sqLoss}
Consider a model in which the evaluator chooses an action $a$ given the realization of the agent's covariates, and  the evaluator and agent share the payoff function $-(a-y)^2$. The following result shows that the conclusion of Part (a) of Theorem \ref{prop:PreferHA} extends for non-binary types $y$.
\begin{proposition} There exists an $N$ sufficiently large such that the agent prefers the black box evaluator for all $n\geq N$.  
\end{proposition}

\begin{proof} 
Throughout the proof  set $s=0$, $\mathbb{E}[Y]=0$ and $\sigma^2=\E(Y_i^2)=1$ without loss.  We will show  that 	\begin{align*}
\E[u(Z_B^n,y)]&-\E\left[\max_{1\le k\le K_n} u(Z_k,y)\right]\\
		&= \E[u(Z_B^n,y)-u(0,y)] - \E\left[\max_{1\le k\le K_n} u(Z_k,y)-u(0,y)\right]>0
	\end{align*}
for large enough $n$.

Let $x_B = (x_i)_{i \in B}$ denote the covariates that Black Box observes, and as before let $Z_B^n = \mathbb{E}[y \mid x_B]$ denote Black Box's (random) posterior expectation.  The optimal action choice $a=Z_B^n$ yields expected payoff $\Var(y \mid x_B)$.  By the Law of Total Variance,
$\mathbb{E}[-\Var(y\mid x_B)] = \Var(Z_B^n)-\Var(Y)$. Since additionally $\mathbb{E}[u(0,y)] = \Var(y)$, we obtain  
    \[\E[u(Z_B^n,y)-u(0,y)] = \mathbb{E}\left[(Z_B^n)^2\right] = \frac{1}{2^{(1-\alpha_B)n}}.\]
		
	Now turn to $\E\left[\max_{1\le k\le K_n} u(Z_k,y)-u(0,y)\right]$. By Lipschitz continuity of $u$, there is a constant $c_2$ such that 
	$u(z_k,y)-u(0,y)\le c_2|z_k|$ holds 
 pointwise for each realization of $(z_k,y)$. So
	\[\E\left[\max_{1\le k\le K_n}u(Z_k,Y)-u(0,Y)\right]\le c_2\E[\max\{|Z_k|\}]\]
The remainder of the proof proceeds identically to the proof of Theorem \ref{prop:PreferHA}.
\end{proof}

\section{Proofs for Results in Sections \ref{sec:Extensions} and \ref{sec:ExtendBreak}}

\subsection{Proof of Corollary \ref{corr:Max}}

We continue in the general setting outlined in the proof of Theorem \ref{thm:MainGeneral}. Fix any realization $\bold{x}_{\mathcal{S}}=(x_1, \dots, x_s)$ of the standard covariates. As in the proof of Theorem \ref{thm:Main}, there are $2^{n-s}$ covariate vectors $\bold{x}_n \in \{0,1\}^n$ with positive probability conditional on $\bold{x}_{\mathcal{S}}$. Index these by $j=1, \dots, 2^{n-s}$, and define
\[Y^{\bold{x}_{\mathcal{S}}}_j \equiv \mathbb{E}_{P^n}\left[Y \mid (X_1, \dots, X_n) = \bold{x}^j_n\right]\]
to be the expected type given covariate vector $\bold{x}_n^j$. 
For each covariate vector $\bold{x}_n$ and each disclosure set $D_k \subseteq \{s+1, \dots, n\}$, there is a corresponding set of covariate vectors $S_k$ such that the evaluator's posterior expectation after the agent discloses his covariates in set $D_k$ is
 \[Z_k^{\bold{x}_{\mathcal{S}}} = \frac{\sum_{j\in S_k}Y^{\bold{x}_{\mathcal{S}}}_j}{|S_k|}.\]
Different from the proof of Theorem \ref{thm:Main}, there are now $\overline{K}_n = \sum_{j=0}^{h_n} {n-s \choose j} 2^j$ unique sets $S_k$ (ranging over not only the different possible sets of covariates to disclose but also their values). By the Binomial Theorem,
\[\sum_{j=0}^{h_n} {n-s \choose j} 2^j \leq \sum_{j=0}^{n-s} {n-s \choose j} 2^j = 3^{n-s}.\]
Following the proof of Lemma \ref{prop:Z}, we obtain that
\[\mathbb{E}\left(\max_{1 \leq k \leq \overline{K}_n} \vert Z_k^{\bold{x}_{\mathcal{S}}} - \mu\vert\right)  \leq \frac{1}{2^{n-s-h_n}} C\sqrt{\log(\overline{K}_n)} \leq \frac{1}{2^{n(1-\alpha_h)-s}} C\sqrt{\log(3^{n-s})}\]
which again converges to zero by assumption that $\alpha_h < 1$. Finally  observe that
\begin{align*}
\mathbb{E}\left[ \max_{\bold{x}_{\mathcal{S}} \in \{0,1\}^s} \left(\max_{1 \leq k \leq K_n} \vert Z_k^{\bold{x}_{\mathcal{S}}} - \mu\vert\right) \right] &\leq \mathbb{E}\left[ \sum_{\bold{x}_{\mathcal{S}} \in \{0,1\}^s} \max_{1 \leq k \leq K_n} \vert Z_k^{\bold{x}_{\mathcal{S}}} - \mu \vert \right] \\
& = \sum_{\bold{x}_{\mathcal{S}} \in \{0,1\}^s}  \mathbb{E}\left[ \max_{1 \leq k \leq K_n} \vert Z_k^{\bold{x}_{\mathcal{S}}} - \mu \vert \right].
\end{align*}
Since each $\mathbb{E}\left[ \max_{1 \leq k \leq K_n} \vert Z_k^{\bold{x}_{\mathcal{S}}} \vert \right] \rightarrow 0$ as $n\rightarrow \infty$, the RHS converges to zero. We thus obtain the analogue of Lemma \ref{lem_abs} for the expected maximum value of context, and the remainder of the proof proceeds identically to Theorem \ref{thm:Main}.

\subsection{Proof of Proposition \ref{prop:DisclosureBound}}

Throughout this proof, we set $s=0$ for simplicity of notation.   

Let $(\sigma^*,\mu^*)$ denote a typical PBE, where $\sigma^*$ is the Sender's disclosure strategy and $\mu^*$ is the Receiver's belief function. Fixing any such equilibrium, we use $Z_{\mu^*}(d)$ to denote the Receiver's posterior expectation given disclosure $d$.  We first prove that at least one pure-strategy equilibrium always exists.

\begin{proposition}
    For every n and f there exists a pure-strategy f-context equilibrium.
\end{proposition}

\begin{proof}
Consider a candidate equilibrium  $(\sigma^*,\mu^*)$, where $\sigma^*(\mathbf{x}_n)=\varnothing$ for all $\textbf{x}_n\in\{0,1\}^n$ (which is 
 clearly a feasible disclosure for all agents). The Receiver's beliefs at disclosure $\varnothing$ are pinned down by Bayes' rule. For any other disclosure $d\neq \varnothing$, we construct out-of-equilibrium beliefs such that $u(Z_{\mu*}(\varnothing)) \geq u(Z_{\mu^*}(d))$. This is always possible, for example by setting $Z_{\mu*}(\varnothing)=  Z_{\mu^*}(d)$ for every $d$. Then by construction reporting $\varnothing$ is a best response for any $\textbf{x}_n$, so we are done.
\end{proof}

  Consider any function $f$ and any pure-strategy equilibrium $(\sigma^*,\mu^*)$ of the $f$-context disclosure game.  Let  $d_1, \dots, d_N$ index the  disclosures that have positive probability under $\sigma^*$ (i.e., all $d \in \mathcal{D}$ such that $\sigma^*(\bold{x}_n)=d$ for some $\bold{x}_n$). For each such disclosure $d_i$,
\[Z_{\mu^*}(d_i) = \frac{1}{\vert \{x : \sigma^*(x)=d_i\}\vert} \sum_{x : \sigma^*(x)=d_i} f(x)\]
is the evaluator's posterior expectation upon observing disclosure $d_i$. Given the evaluator's payoff function, the optimal action for the evaluator is precisely $Z_{\mu^*}(d_i)$.
Let 
\begin{equation} \label{eq:DisclosureBest}
d^*=\left(H^*,(\mathcal{X}^*_i)_{i \in H^*}\right):=\arg\max_{1 \leq i \leq N} u(Z_{\mu^*}(d_i))
\end{equation}
be the disclosure that yields the highest  payoff to the Sender.  Then it must be that $\sigma^*(\bold{x}_n)=d^*$ for every covariate vector $\bold{x}_n$ for which disclosure $d^*$ is feasible. Otherwise $d^*$ would be a profitable deviation. Hence the evaluator's posterior expectation in this equilibrium is the same as it would have been given disclosure of $d^*$ in our main model. So 
\[u(Z_{\mu^*}(d^*)) \leq \max_{\bold{x}_n \in \{0,1\}^n}v(f,\bold{x}_n).\]
Since the payoff received by an agent with any other covariate vector cannot exceed $u(Z_{\mu^*}(d^*))$ (by (\ref{eq:DisclosureBest})), we have the desired result.

\subsection{Result for Mixed Strategy Equilibria}

In this part we restrict to equilibria $(\sigma^*,\mu^*)$ with the property that
$\argmax_{\hat{y} \in A_{(\sigma^*,\mu^*)}}  u(\hat{y})$ is unique on the set $A_{(\sigma^*,\mu^*)}$ of  posterior expectations with positive probability in this equilibrium.  Call these equilibria \emph{generic}. (A sufficient condition for all equilibria to be generic  is if $u$ is strictly monotone.)

For each $n$ and $f$, let $v^D(f,\bold{x}_n)$ denote the highest payoff that an agent with covariate vector $\bold{x}_n$ receives in any generic 
equilibrium (potentially mixed) of the $f$-context disclosure game. Further define
\[v^D_f(n) = \max_{\bold{x}_n} v^D(f,\bold{x}_n)\]
and
\[V^{\mathcal{D}}(n)=\E[v^D_f(n)]\]
where the expectation is with respect to the realization of $f$. 	

 \begin{proposition}
     Suppose Assumption \ref{assp:Exchangeability}   holds and $u(.)$ is twice continuously differentiable. Then $\lim_{n \rightarrow \infty} V^{\mathcal{D}}(n) = 0$.
 \end{proposition}
	
	\begin{proof}
		
		Fix $n$, $f$, and a context equilibrium $(\sigma^*,\mu^*)$ of the $f$-context disclosure game. Let $\mathcal{Z}^*\subseteq [-\overline{y},\overline{y}]$ be the compact set of all  equilibrium posterior expectations that are realized  with positive probability in this equilibrium. Further, denote 
		\[Z_{(1)}^*=\arg \max_{z\in \mathcal{Z}^*}u(z)\]
		 to be the most-preferred achievable posterior expectation, which is unique by assumption of genericity of the equilibrium. 
		
		Since $Z_{(1)}^*$ is the best attainable posterior expectation, an agent achieves  $Z_{(1)}^*$ in equilibrium if and only if it is feasible. (Otherwise, the agent can profitably deviate to the feasible disclosure that induces this posterior expectation.) 
  
  Let $\mathcal{X}^* \subseteq \{0,1\}^n$
		denote the set of agents who have a feasible disclosure that achieves $Z_{(1)}^*$. Let $\mathcal{D}(\mathcal{X}^*)$  be the set of disclosures that agents in $\mathcal{X}^*$ send with positive probability in equilibrium. By the logic above, $\mathcal{D}(\mathcal{X}^*)\cap \mathcal{D}(\mathcal{X}\setminus \mathcal{X}^*)=\varnothing$. Using the structure of this equilibrium we can write
  \begin{equation}
  \E[Y]=Z_{(1)}^*p_{\mathcal{X}^*}+(1-p_{\mathcal{X}^*})\E[Y|X \notin \mathcal{X}^*]
  \label{eqn:bp_eq}
  \end{equation}
  where $p_{\mathcal{X}^*}$ is the ex-ante probability that the agent's covariate vector belongs to $\mathcal{X}^*$, and $\E[Y|X \notin \mathcal{X}^*]$ is the expectation of the agent's type given that his covariate vector does not belong to $\mathcal{X}^*$.
Here  we utilize the fact that the evaluator's posterior expectation is constant at $Z_{(1)}^*$ across all agents with covariate vectors in $\mathcal{X}^*$.\footnote{In general this does not have to be the case. We rule this out in the definition of the equilibrium.} 
		
		Now, consider the following alternative ``strategy'' $\sigma_0$, which relaxes the feasibility constraint: For any $\bold{x}\in \mathcal{X}\setminus \mathcal{X}^*$  let $\sigma_0(\textbf{x})\equiv \sigma^*(\textbf{x})$, i.e., the disclosures are the same as in the original equilibrium. Further choose some arbitrary disclosure $d_0\in \mathcal{D}(\mathcal{X}^*)$ and  let $\sigma_0(\textbf{x})=d_0$  for all $\bold{x}\in \mathcal{X}^*$. The Receiver's posterior expectation following observation of disclosure $d_0$ is
		\[Z_0=\frac{\sum_{x\in \mathcal{X}^*}Y_x}{|\mathcal{X}^*|}\]
		and, analogous to (\ref{eqn:bp_eq}), we can write
   \begin{equation}
\E[Y]=Z_{0}p_{\mathcal{X}^*}+(1-p_{\mathcal{X}^*})\E[Y|X \notin \mathcal{X}^*]
  \label{eqn:bp_str}
  \end{equation}
Combining equations (\ref{eqn:bp_eq}) and (\ref{eqn:bp_str}) we conclude:
		\[Z_{(1)}^*=\frac{\sum_{x\in \mathcal{X}^*}Y_x}{|\mathcal{X}^*|}\]		
		which almost surely converges to $\E[Y]$ so long as $|\mathcal{X}^*|\xrightarrow{n\rightarrow\infty}\infty$. Since the $Y_x$'s are uniformly bounded, this also implies $\mathbb{E}[Z^*_{(1)}] \rightarrow \mathbb{E}[Y]$, as desired. We now  demonstrate that indeed  $|\mathcal{X}^*|\xrightarrow{n\rightarrow\infty}\infty$.	
  
  For any disclosure $d$ denote by $C_{d}\subseteq \{0,1\}^n$ the set of all covariate vectors $\textbf{x}$ given which $d$ is feasible. 
  Since $Z_{(1)}^*$ is achieved by all agents for whom $Z_{(1)}^*$ is feasible, it must be that for every disclosure $d\in \mathcal{D}(\mathcal{X}^*)$ 
  we have  $C_{d}\subseteq \mathcal{X}^*$. Then for any $d \in \mathcal{D}(\mathcal{X}^*)$,
		\[|\mathcal{X}^*|\ge |C_{d}|\xrightarrow{n\rightarrow\infty}\infty.\] where the limit follows by assumption that $\alpha_h<1$. 	This completes the proof.	
	\end{proof}

\subsection{Proof of Proposition \ref{prop:Effective}}
We again continue in the general setting outlined in the proof of Theorem \ref{thm:MainGeneral}, and adopt the conventions that $\mathbb{E}(Y)=\mu$ while $\Var(Y) = 1$. We prove the result for a weakening of Assumptions \ref{assp:ExchangeabilityRelevant} and \ref{assp:ConstantVarYRelevant} to the following.

\begin{assumption} \label{assp:Irrelevant} Fix any realization of the standard covariates $\bold{x}_{\mathcal{S}} \in \{0,1\}^s$. There is an infinitely exchangeable sequence $(\widetilde{Y}_1, \widetilde{Y}_2, \dots)$ such that for every $n\in \mathbb{N}$, the sequence 
\[(Y_{\bold{x}_{R_n}, \bold{x}_{-R_n}}: (x_i)_{i \in R_n \backslash \mathcal{S}} \in \{0,1\}^{r_n-s})\]
has the same distribution as $(\widetilde{Y}_1, \dots, \widetilde{Y}_{2^{r_n}})$. 
\end{assumption}

Recalling that $r_n$ is the number of relevant covariates, there are $2^{r_n}$ distinct expected conditional types, which we can enumerate as $Y_1, \dots, Y_{2^{r_n}}$. If disclosure $k$ involves disclosing $k_r$ relevant covariates, then there is a set $S_k$ of size $2^{r_n-k_r}$ such that the evaluator's posterior expectation can be written
\[Z_k = \frac{1}{2^{n-h_n}} \sum_{j \in S_k} 2^{n-r_n - (h_n-k_r)} Y_j = \frac{1}{2^{r_n-k_r}} \sum_{j \in S_k} Y_j.\]
As in Step 1 of the proof of Theorem \ref{thm:Main} (Section \ref{proof:Step1}), replace each $Y_j$ with a variable $Y_j^k \stackrel{d}{=} Y_j$ which is independent across disclosure sets. This yields the random variables
\[Z_k^{ind} = \frac{1}{2^{r_n-k_r}} \sum_{j \in S_k} Y^k_j.\]
As in the proof of Proposition \ref{prop:Ind}, it follows from Lemma \ref{fosd_g} that 
\[\mathbb{E}[\max\{Z_1, \dots, Z_{K_n}\}] \leq \mathbb{E}[\max\{Z_1^{ind}, \dots, Z_{K_n}^{ind}\}].\]

Next define 
\[Z^{iid}_k = \frac{1}{2^{r_n-h_n}} \sum_{j=1}^{2^{r_n-h_n}} Y_j^k\]
and note that these are identically and independently distributed with shared variance
\[Var(Z^{iid}_k) = \frac{1}{2^{r_n-h_n}}.\]
Following the arguments in Step 2 of the proof of Theorem \ref{thm:Main} (Section \ref{proof:Step2}), we get
\[\mathbb{E}[\max\{Z^{ind}_1, \dots, Z^{ind}_{K_n}\}] \leq \mathbb{E}[\max\{Z^{iid}_1, \dots, Z^{iid}_{K_n}\}].\]
where as before $K_n=\sum_{j=0}^{h_n}\binom{n}{j}$. Further, by the argument given in Step 3 of the proof of Theorem \ref{thm:Main} (Section \ref{proof:Step3}),
\[\lim_{n \rightarrow \infty} \vert V^{iid}_n - V^N_n \vert = 0\]
where
\[V^N_n\equiv \E\left[\maxx{Z_1^N, \dots, Z_{K_n}^N}\right]\]
and $Z_k\sim \mathcal{N}\left(\mu,\frac{1}{2^{r_n-h_n}}\right)$. Again applying the bound from \citet{Berman1964}, we have
\[V^N_n\le \frac{1}{2^{r_n-h_n}} C\sqrt{\log(K_n)}\le\frac{1}{2^{n(\alpha_r-\alpha_h)}} C\sqrt{n}.\]
By assumption that $\alpha_r > \alpha_h$, the right-hand expression converges to zero as $n$ grows large, concluding the proof.

\subsection{Supporting Materials for Section \ref{sec:NonConstant}}
\label{app:accum}

We show here that both main results generalize to the setting with noise. Instead of repeating the proof step by step, we emphasize what changes must be made in order for the proofs to translate. To be consistent with our previous notation, we again operate  with $\Tilde y_{i,n}=Y_i+\varepsilon_n$ where $i$ enumerates covariate vectors. In all subsequent notation a tilde will indicate an object that includes noise $\varepsilon_n$, and objects without one are the same as in the main text. 
Observe that after adding the noise term $\varepsilon_n$,  the key objects from the proofs transform in the following way: $Z_k$ is replaced by
$\Tilde{Z}_k=Z_k+\varepsilon_n$, and 
$\E[\max_{k}Z_k]$ is replaced by $\E[\max_{k}\Tilde{Z}_k]$. Since $\E[\max_{k}\Tilde{Z}_k]=\E[\max_{k}Z_k]$, the proof of Theorem \ref{thm:Main} translates directly.

To demonstrate Theorem \ref{prop:PreferHA}, we will show how to adjust the proof of part (a) with part (b) following analogously. The first change is that (\ref{eq:LB}) becomes
\[\E[\phi(\Tilde Z_B^n)-\phi(0)]\ge c_1\E[(Z_B^n)^2]-c_1 \sigma_{\varepsilon,n}^2\]
The inequality in (\ref{eq:UB}) is also modified to
\[\phi(\Tilde{Z}_k)-\phi(0)\le c_2|Z_k|+c_2|\varepsilon_n|\]
for every disclosure $k$. Thus we obtain
\[\Delta(n)\ge c_1\frac{1}{2^{(1-\alpha_b)n}}-2c_2\frac{1}{2^{(1-\alpha_h)n}}\sqrt{\log(K_n)}+c_2\E[|\varepsilon_n|]-c_1 \sigma_{\varepsilon,n}^2\]

We will show that the ratio $\frac{\E[|\varepsilon_n|]}{\sigma_{\varepsilon,n}^2}$ grows arbitrary large with $n$, thus asymptotically exceeding $\frac{c_1}{c_2}$. Fixing some $d>0$ and applying Markov inequality, we obtain
\begin{equation}
    \frac{\E[|\varepsilon_n|]}{\sigma_{\varepsilon,n}^2}\ge d\cdot \P\left(\frac{|\varepsilon_n|}{\sqrt{\sigma_{\varepsilon,n}^2}}\ge d\sqrt{\sigma_{\varepsilon,n}^2}\right)
    \label{ineq:Markov}
\end{equation}
If we denote the CDF of $\frac{\varepsilon_n}{\sqrt{\sigma_{\varepsilon,n}^2}}$ as $G_n$, the RHS of the above inequality can be rewritten as $d\cdot \left(1+G_n(-d\sqrt{\sigma_{\varepsilon,n}^2}))-G_n(d\sqrt{\sigma_{\varepsilon,n}^2})\right)$. As $n$ grows large, the term in brackets tends to $1-2g_n(0)d\sqrt{\sigma_{\varepsilon,n}^2}+o\left(\sqrt{\sigma_{\varepsilon,n}^2}\right)$. We will omit the $o(\cdot)$ term until the end of the proof.

Fix an arbitrary $\delta>0$ and let $d=\frac{c_1}{c_2}+2\delta$. Further, fix $N$ such that $\sqrt{\sigma_{\varepsilon,n}^2}\le \frac{1}{2gd}(1-\frac{c_1+c_2\delta}{c_1+2c_2\delta})$, where $g=\max_n g_n(0)<\infty$. Then since $\sigma_{\varepsilon,n}^2$ is decreasing and $g_n(0)\le g$ we have that for all $n\ge N$
\[2g_n(0)d\sqrt{\sigma_{\varepsilon,n}^2}\le 2gd\frac{1}{2gd}\left(1-\frac{c_1+c_2\delta}{c_1+2c_2\delta}\right)\]
Combining this inequality with (\ref{ineq:Markov}) we get
\[\frac{\E[|\varepsilon_n|]}{\sigma_{\varepsilon,n}^2}\ge \left(\frac{c_1}{c_2}+2\delta\right)\left(1-\left(1-\frac{c_1+c_2\delta}{c_1+2c_2\delta}\right)\right)+o\left(\sqrt{\sigma_{\varepsilon,n}^2}\right)=\left(\frac{c_1}{c_2}+\delta\right)+o\left(\sqrt{\sigma_{\varepsilon,n}^2}\right)\]
for all $n \geq N$. Since $\delta$ is an arbitrary positive number, this concludes the proof.

\subsection{Proof of Proposition \ref{prop:non-bayesian}}

Throughout the proof we assume $u(x)\equiv x$ and $s=0$.  In addition, for simplicity of notation, we enumerate feasible disclosures by $k$
 and denote the corresponding posteriors  (as random variables) as 
$Z_k^n:=\rho_f(d_k)$.  To upper bound the value of context, we apply a result from \cite{arnold1979bounds}: 
	\begin{equation}
 \begin{split}
	&\left|\E\left[\max_{k\in\{1,...,K_n\}}Z_k^n-\E\left[\frac{\sum_{i=1}^{K_n}Z_i^n}{K_n}\right]\right]\right|\le 
	  \\
&   
	\sqrt{\left(1-\frac{1}{K_n}\right)\sum_{i=1}^{K_n}Var(Z_i^n)+
	  \frac{1}{K_n}\sum_{i=1}^{K_n}\left(\sqrt{K_n}\left(\E[Z_i^n]-\frac{\sum_{i=1}^{K_n}\E[Z_i^n]}{K_n}\right)\right)^2} 
   \label{eq:non-bayes_upper}
   \end{split}
   \end{equation}
By Assumption \ref{assp:Unbiased}, inequality \ref{eq:non-bayes_upper} simplifies to 
\[\left|\E\left[\max_{k\in\{1,...,K_n\}}Z_k^n\right]-\mu \right|\le 
	\sqrt{\left(1-\frac{1}{K_n}\right)\sum_{i=1}^{K_n}Var(Z_i^n)}\]
		 Finally, Assumption \ref{assp:Concentrate} implies that  $Var(Z_k^n)=o(\frac{1}{K_n})$ for every disclosure $k$. Hence
   
\[\left|\E\left[\max_{k\in\{1,...,K_n\}}Z_k^n\right]-\mu \right|\le 
	\sqrt{\left(1-\frac{1}{K_n}\right)K_no(K_n^{-1})}\]
 which yields the desired result after taking a limit in $n$. The argument for the lower bound  follows the same line of reasoning and is thus omitted.

\end{document}